\documentclass[format=acmsmall,review=false]{acmart}

\usepackage{acm-ec-19}

\usepackage{microtype}

\usepackage{booktabs}

\usepackage{bm}
\usepackage{color}
\usepackage{url}
\usepackage{graphicx}
\usepackage{pgfplots}
\usepackage{tikz}
\usepackage{algorithm}
\usepackage{algpseudocode}
\usepackage{mathrsfs}
\usepackage{relsize}
\usepackage{yfonts}
\usepackage{nicefrac}
\usepackage{multicol}
\usepackage{subcaption}
\usepackage{wrapfig}
\usepackage{sgame}

\usepackage{mathrsfs}
\usepackage{verbatim}
\usepackage{cleveref}
\usepackage{microtype}
\usepackage[normalem]{ulem}

\usepackage{cleveref}

\usepackage{amsmath,amssymb,amsfonts,amsthm}
\usepackage{mathtools}
\usepackage{bbm}
\usepackage{dsfont}

\usetikzlibrary{arrows, automata}

\pgfplotsset{compat=1.14}

\newcommand{\textcomp}[1]{\textgoth{#1}}

\newcommand{\RCOP}{\textcomp{R}}
\newcommand{\ERCOP}{\hat{\RCOP}}
\newcommand{\DERCOP}{\ERCOP{}^1}

\newcommand{\RC}[3]{\RCOP_{\mathsmaller{#1}}\!\left(#2,\, #3\right)}

\newcommand{\DERC}[4]{\DERCOP_{\mathsmaller{#1}}\!\left(#2,\, #3,\, #4\right)}

\newcommand{\probop}[1]{\mathop{\mathbb{#1}}\displaylimits}

\newcommand{\Expop}{\probop{E}}

\newcommand{\Expwrt}[2]{\Expop_{#1}\left[#2\right]}

\DeclareMathOperator{\support}{support}

\newcommand{\rand}[1]{\mathscr{#1}}

\newcommand{\iid}{\text{i.i.d.}}
\newcommand{\distributed}{\sim}

\newcommand{\Rademacher}{\mathsmaller{\operatorname{Rademacher}}}

\newcommand{\lnp}[1]{\ln\!\left(#1\right)}

\newcommand{\abs}[1]{\left\lvert{} #1 \right\rvert{}}
\newcommand{\norm}[1]{\left\lVert{} #1 \right\rVert{}}

\DeclareMathOperator*{\argmax}{arg\,max}

\newcommand{\R}{\mathbb{R}}
\newcommand{\N}{\mathbb{N}}

\newcommand{\vsigma}{\bm{\sigma}}

\newcommand{\MDOP}{\textcomp{D}}
\newcommand{\EMDOP}{\hat{\MDOP}}

\newcommand{\EMD}[3]{\EMDOP_{\mathsmaller{#1}}\!\left(#2,\,#3\right)}

\newcommand{\MD}[3]{\MDOP_{\mathsmaller{#1}}\!\left(#2,\,#3\right)}

\newif\iflsymb
\lsymbfalse

\iflsymb

\else

\fi

\allowdisplaybreaks

\newcommand{\slfrac}[2]{#1/#2}

\newcommand{\sabs}[1]{\lvert{} #1 \rvert{}}

\newcommand{\mydef}[1]{\textbf{#1}}

\newcommand{\amy}[1]{{\color{red}[\textsc{Amy}: \emph{#1}]}}

\newfloat{algorithm}{t}{lop}
\floatname{algorithm}{Algorithm}
\algnewcommand{\Input}{\textbf{input:}~}
\algnewcommand{\Output}{\textbf{output:}~}
\algnewcommand{\Error}{\textbf{error}~}
\algnewcommand{\Continue}{\textbf{continue}~}
\algnewcommand{\Break}{\textbf{break}~}

\newcommand{\Samples}{\bm{X}}
\newcommand{\SamplePoint}{x} %

\newcommand{\NumberOfSamples}{m}

\newcommand{\NoiseCondition}{d}
\newcommand{\UtilityRange}{c}
\newcommand{\ConditionSpace}{\mathcal{X}}
\newcommand{\ConditionValue}{x}

\newcommand{\ConditionDistribution}{\rand{D}}
\newcommand{\SampleIndex}{j}

\newcommand{\TimeIndex}{t}

\newcommand{\SamplingSchedule}{\bm{M}}

\newcommand{\GameTuple}{\Gamma}
\newcommand{\ConditionalGame}[1]{\GameTuple_{#1}}    %
\newcommand{\InducedGame}[1]{\GameTuple_{#1}}        %
\newcommand{\EmpiricalGame}[1]{\hat{\GameTuple}_{#1}} %

\newcommand{\NumberOfPlayers}{|\SetOfPlayers|}
\newcommand{\SetOfPlayers}{P}
\newcommand{\PlayerIndex}{p}
\newcommand{\StrategySet}{S}
\newcommand{\StrategySetAlt}{T}
\newcommand{\Strategy}{s}

\newcommand{\StratProfile}{\bm{s}}
\newcommand{\StratProfileAlt}{\bm{t}}
\newcommand{\StratProfileSpace}{\bm{\StrategySet}}

\newcommand{\MixedStrategySet}{{\StrategySet}^{\diamond}}
\newcommand{\MixedStrategySetAlt}{{\StrategySetAlt}^{\diamond}}

\newcommand{\MixedStratProfileAlt}{\bm{\tau}}
\newcommand{\MixedStratProfileSpace}{\bm{\MixedStrategySet}}
\newcommand{\MixedStratProfileSpaceAlt}{\bm{\MixedStrategySetAlt}}

\newcommand{\Utility}{\bm{u}}

\newcommand{\BRG}{B}

\newcommand{\Neighborhood}{\mathcal{N}}

\newcommand{\SizeOfGame}[1]{N_{#1}}

\DeclareMathOperator{\Nash}{E}

\DeclareMathOperator{\Regret}{Reg}
\DeclareMathOperator{\Rationalizable}{Rat}

\DeclareMathOperator{\restrict}{restrict}

\DeclareMathOperator{\Image}{Image}

\newcommand{\GS}{\ensuremath{\operatorname{GS}}}
\newcommand{\PSP}{\ensuremath{\operatorname{PSP}}}
\newcommand{\Indices}{\bm{I}}

\newcommand{\SetOfFacilities}{E}
\newcommand{\NumberOfFacilities}{|\SetOfFacilities|}
\newcommand{\FacilityIndex}{e}
\newcommand{\FacilityCostFunction}{f}
\newcommand{\PlayerCost}{C}
\newcommand{\CongestionGame}{\mathcal{C}}

\newcommand{\CongestionSetOfStrategies}{\StrategySet}

\newcommand{\RandomGame}{\mathrm{RG}}
\newcommand{\RandomCongestionGame}{\mathrm{RC}}

\newcommand{\smallsup}[1]{\mathrel{\raisebox{0.25ex}{\ensuremath{\displaystyle\sup_{\mathsmaller{#1}}}}}}

\newtheorem{observation}{Observation}[section]

\begin{document}

\title[Learning games from data]{Learning Equilibria of Simulation-Based Games}
\author{Enrique Areyan Viqueria, Cyrus Cousins, Eli Upfal, Amy Greenwald}

\begin{abstract}
We tackle a fundamental problem in empirical game-theoretic analysis (EGTA), that of learning equilibria of simulation-based games. Such games cannot be described in analytical form; instead, a black-box simulator can be queried to obtain noisy samples of utilities. Our approach to EGTA is in the spirit of probably approximately correct 
learning. We design algorithms that learn so-called empirical games, which uniformly approximate the utilities of simulation-based games with finite-sample guarantees.
These algorithms can be instantiated with various concentration inequalities.
Building on earlier work, we first apply Hoeffding's bound, but as the size of the game grows, this bound eventually becomes statistically intractable; hence, we also use the Rademacher complexity.
Our main results state: with high probability, all equilibria of the simulation-based game are approximate equilibria in the empirical game (perfect recall); and conversely, all approximate equilibria in the empirical game are approximate equilibria in the simulation-based game (approximately perfect precision).
We evaluate our algorithms on several synthetic games, showing that they make frugal use of data, produce accurate estimates more often than the theory predicts, and are robust to different forms of noise.
\end{abstract}

\maketitle\theoremstyle{remark}

\section{Introduction}

Game theory is the \emph{de facto\/} standard conceptual framework used to analyze the strategic interactions among 
rational agents in multi-agent systems. At the heart of this framework is the theoretical notion of a game. In a game, each player (henceforth, an agent) chooses a strategy 
and earns a utility that in general depends on the profile (i.e., vector) of strategies chosen by all the agents.

The most basic representation of a game is the \mydef{normal form}, which can be visualized as a matrix containing all agents' utilities at all strategy profiles.
There are many techniques available to analyze normal-form games.
Analyzing a game means predicting its outcome, i.e., the strategy profiles that one can reasonably expect the agents to play.
Perhaps the most common such prediction
is that of \mydef{Nash equilibrium}~\cite{nash1950equilibrium}, where each agent plays a \mydef{best-response} to the strategies of the others: i.e., a strategy that maximizes their utility.
More generally, at an $\epsilon$-Nash equilibrium agents best respond to other agents' strategies up to an additive error of $\epsilon$.

This paper is concerned with analyzing games for which a complete and accurate description is not available.
While knowledge of the number of agents and their respective strategy sets is available, we do not assume {\it a priori} access to the game's utility function.
Instead, we assume access to a 
simulator from which we can sample noisy utilities associated with any strategy profile.
Such games have been called \mydef{simulation-based games}~\cite{vorobeychik2008stochastic} and
\mydef{black-box games}~\cite{picheny2016bayesian},
and their analysis is called \mydef{empirical game theoretic analysis} (EGTA)~\cite{wellman2006methods,jordan2007empirical}.
EGTA methodology has been applied in a variety of practical settings for which simulators are readily available, including trading agent analyses in supply chain management~\cite{vorobeychik2006empirical,jordan2007empirical}, ad auctions~\cite{jordan2010designing}, and
energy markets~\cite{ketter2013autonomous};
designing network routing protocols~\cite{wellman2013analyzing}; strategy selection in real-time games~\cite{tavares2016rock}; and the dynamics of reinforcement learning algorithms, such as AlphaGo \cite{tuyls2018generalised}.

The aim of this work is to design learning algorithms that can accurately estimate all the equilibria of simulation-based games. We tackle this problem using the \mydef{probably approximately correct} (PAC) learning framework~\cite{valiant1984theory}.
Our algorithms learn so-called \mydef{empirical games}~\cite{wellman2006methods}, which are estimates of simulation-based games constructed via sampling.
We argue that empirical games so constructed yield \mydef{uniform approximations} of simulation-based games, meaning all utilities in the empirical game tend toward their expected counterparts,
\emph{simultaneously}.

This notion of uniform approximation is central to our work: we prove that, when one game $\GameTuple$ is a uniform approximation of another $\GameTuple'$, all equilibria in $\GameTuple$ are approximate equilibria in $\GameTuple'$.
In other words, a uniform approximation implies perfect recall and \emph{approximately\/} perfect precision: all true positives in $\GameTuple$ are approximate equilibria in $\GameTuple'$, and all false positives in $\GameTuple'$ are approximate equilibria in $\GameTuple$.
Our learning algorithms, which learn empirical games that are uniform approximations of simulation-based games, thus well estimate the equilibria of simulation-based games.

Estimating \emph{all\/} the utilities in an empirical game is non-trivial, in part because of the \mydef{multiple comparisons problem} (MCP), which arises when estimating many parameters simultaneously, since accurate inferences for each individual parameter do not necessarily imply a similar degree of accuracy for the parameters in aggregate.  Controlling the \mydef{family-wise error rate} (FWER)---the probability that one or more of the approximation guarantees is violated---is one way to control for multiple comparisons.
In this work, we consider two approaches to FWER control in empirical games.
The first builds on classical methods, using concentration inequalities to first establish confidence intervals about each parameter \emph{individually}, and then applying a statistical correction (e.g., Bonferroni or \v{S}id\'{a}k) to bound the probability that all 
approximation guarantees hold \emph{simultaneously}, often incurring
looseness in one or both steps.
The second is based on Rademacher averages, which directly bound the error of all parameters simultaneously, by replacing the per-parameter concentration inequalities with a single concentration inequality, and an additional data-dependent term, thereby circumventing the looseness inherent in classical methods.

In addition to controlling the FWER in two ways, we also consider two learning algorithms.
The first, which we call \mydef{global sampling} (\GS), learns an empirical game from a static sample of utilities.  The second, \mydef{progressive sampling with pruning} (\PSP), samples dynamically.
Prior art~\cite{elomaa2002progressive,riondato2015mining, riondato2018abra} uses progressive sampling to obtain a desired accuracy at a fixed failure probability by guessing an initial sample size, computing Rademacher bounds, and repeating with larger sample sizes until the desired accuracy is attained.
Our work enhances this strategy with \emph{pruning}: at each iteration, parameters that have already been sufficiently well estimated for the task at hand (here, equilibrium estimation) are pruned, so that subsequent iterations of \PSP{} do not refine their bounds.  Pruning is compatible with both classical MCP correction methods and with Rademacher averages.
In both cases, pruning represents both a statistical and computational improvement over earlier progressive sampling techniques.

We empirically evaluate both \GS{} and \PSP{} with a Bonferroni correction on randomly generated games and finite congestion games~\cite{rosenthal1973class}.
We show that they make frugal use of data, accurately estimate equilibria of games more often than the theory predicts, and are robust to different forms of noise.
We further show that \PSP{} can significantly outperform \GS, the current state-of-the-art~\cite{tuyls2018generalised}, relying on far fewer data to produce the same (and often times, better) error rates.

In the simulation-based games of interest here, a running assumption is that simulation queries are exceedingly expensive, so that the time and effort required 
to obtain sufficiently accurate utility estimates dwarves that of any other relevant computations, including equilibrium computations.  This assumption is reasonable in so-called \mydef{meta-games}~\cite{tavares2016rock,wellman2006methods}, where agents choice sets comprise a few high-level heuristic strategies, not intractably many low-level game-theoretic strategies.  Thus, our primary concern is to limit the need for sampling, while still %
well estimating a game.

\paragraph{Related Work}

The goal of empirical game-theoretic analysis is to solve games that are not directly amenable to theoretical analysis because they can be described only by repeatedly querying a simulator, and even then their description is necessarily noisy.
Most EGTA methods work by first estimating utilities to construct an empirical game, and then computing an approximate Nash equilibrium of that game.
One approach to estimating utilities is to sample a few strategy profiles, and then learn from those observations functions that generalize utilities to unseen data~\cite{vorobeychik2007learning,wiedenbeck2018regression}.
Another line of work proposes sampling heuristics based on, for example, information-theoretic 
principles~\cite{walsh2003choosing}, or variance reduction techniques in games with stochastic outcomes where agents employ mixed strategies~\cite{burch2016aivat}. 
To handle particularly complex games, hierarchical game reduction~\cite{wellman2006methods} can be used to reduce the size of the game. 
Recently, EGTA has been successfully applied in games with very large strategy spaces, like StarCraft~\cite{tavares2016rock}, by restricting their strategy spaces to small sets of predefined strategies (i.e., heuristics)---in other words, by solving them as meta-games.  

One distinguishing feature of our work \emph{vis \`a vis\/} much of the existing EGTA work is that we aim to estimate all equilibria of a simulation-based game, rather than just a single one (e.g., Jordan \emph{et al.}~\cite{jordan2008searching}).
Notable exceptions include Tuyls, \emph{et al.}~\cite{tuyls2018generalised}, Vorobeychik~\cite{vorobeychik2010probabilistic}, and Wiedenbeck \emph{et al.}~\cite{Wiedenbeck:2014:AGT:2615731.2616156}.
The first of these papers shows that all equilibria of a simulation-based game are also approximate equilibria in an empirical game: i.e., they establish perfect recall with finite-sample guarantees.
The second paper derives asymptotic results about the quality of the equilibria of empirical games: i.e., they establish perfect precision in the limit, as the number of samples tends to infinity.
Our main theorem unifies these two results, obtaining finite sample guarantees for both perfect recall and approximate precision.
We obtain our results in a manner that generalizes (in theory, if not in practice) the approach taken by Tuyls, \emph{et al.}~\cite{tuyls2018generalised}, and using Rademacher averages.

The third of the aforementioned works, Wiedenbeck \emph{et al.}~\cite{Wiedenbeck:2014:AGT:2615731.2616156}, is perhaps closest in spirit to our own.
There the goals, like ours, are to characterize the quality of the equilibria in empirical games, and to exploit statistical information to save on sampling.
Their methods employ bootstrapping~\cite{EfroTibs93}; as such, they do not immediately afford any theoretical guarantees.  Still, it could be of interest to test our methods against theirs, under the usual EGTA assumption that little to nothing is known about how the games' utilities are distributed; but at first blush, it seems likely that neither would dominate.  Bootstrapping methods usually yield 
stronger approximation guarantees,
but it is easy to construct examples where they are overly optimistic (e.g., when the sampling distribution is skewed, or in case there are outliers); thus, for applications that require a high degree of certainty, bootstrapping would likely not be the preferred method.

Sample complexity bounds are an essential tool in EGTA, as running the simulator is generally a computational bottleneck. In query complexity work (e.g., \cite{babichenko2016query,fearnley2015learning,jecmen2018bounding}), one seeks to bound the number of samples needed to approximate a mixed strategy Nash equilibrium. 
This line of work differs from our work in at least two critical aspects: (1) sampling is over mixed strategies, not noisy payoffs, and (2) one seeks a single equilibrium. Noisy payoffs may have little impact on their own, but when combined with seeking all equilibria, they are substantial, as we must account for multiple hypothesis testing (motivating our use of uniform convergence bounds). Moreover, since our methods estimate all equilibria, which essentially requires learning the entire game, lower bounds from the query complexity literature do not apply.  This difference is key; not only does estimating all equilibria yield a greater understanding of a game, single equilibria are not sufficient to estimate certain properties, such as the price of anarchy.

Rademacher averages were introduced to the statistical learning theory literature for supervised learning \cite{koltchinskii2001rademacher,bartlett2002rademacher}, but have also been used in unsupervised learning problems like mining frequent itemsets \cite{riondato2015mining}, and sampling problems like betweenness centrality approximation \cite{riondato2018abra}.  
We estimate Rademacher averages using Monte-Carlo simulations, as suggested by Bartlett and Mendelson \cite{bartlett2002rademacher}; doing so, yields very tight generalization bounds.  Although a small number of supervised learning papers, including learning decision tree prunings \cite{kaariainen2004selective},
take this approach, computation of the requisite suprema is often intractable, 
leading most to rely on looser analytical bounds instead.%
\footnote{Loose analytical bounds can be valuable, for the insight they provide into when the sample complexity is high or low.  We pursue this approach in Appendix~\ref{sec:comparingBounds}.}
In our application domain, namely EGTA as opposed to machine learning, simulations dominate run time, so effort that is spent computing more accurate bounds is well compensated for in the simulation savings that ensue.

\paragraph{Outline}
This paper is organized as follows. 
In \Cref{sec:approx},
we present our approximation framework, in which we show that equilibria in games are \emph{approximable}: i.e., they are approximately preserved by uniform game approximations.
In \Cref{sec:learn}, we present our learning framework, in which we show that empirical games uniformly 
converge to their 
expected counterparts with high probability, given finitely many samples.
We conclude that equilibria in simulation-based games are \emph{learnable} from finitely many samples: i.e., approximable with high probability.
In \Cref{sec:algos}, we present two algorithms that uniformly learn simulation-based games; the first is expensive but always
yields guarantees, while the second is more statistically efficient, but can exhaust its sampling budget and fail.
Finally, in \Cref{sec:expts}, we experiment with these algorithms, showing that they can indeed learn equilibria, often more efficiently than the theory suggests.

\section{Approximation Framework}
\label{sec:approx}

We begin by presenting standard game-theoretic notions.
We then introduce the notion of uniform approximation.
Given an approximation of one game by another, there is not necessarily a connection between their properties.
For example, there may be equilibria in one game with no corresponding equilibria in the other, as small changes to the utility functions can add or remove equilibria. 
Nonetheless, we show that finding the approximate equilibria of a uniform approximation of a game is sufficient for finding all the (exact) equilibria of the game itself.

\subsection{Basic Game Theory}

\begin{definition}[Pure Normal-Form Game]
A \mydef{normal-form game} $\GameTuple \doteq \langle \SetOfPlayers, \{ \StrategySet_\PlayerIndex \mid \PlayerIndex \in \SetOfPlayers \},
\Utility(\cdot) \rangle$
consists of a set of agents $\SetOfPlayers$,
with \mydef{pure strategy set}  $\StrategySet_\PlayerIndex$ available to agent $\PlayerIndex \in \SetOfPlayers$.  We define $\StratProfileSpace \doteq \StrategySet_1 \times \dots \times \StrategySet_{\NumberOfPlayers}$ to be the pure strategy profile space of $\GameTuple$, and then $\smash{\Utility : \StratProfileSpace \to \R^{\NumberOfPlayers}}$ is a vector-valued utility function (equivalently, a vector of $\NumberOfPlayers$ scalar utility functions $\Utility_\PlayerIndex$).  
\end{definition}

Given a normal-form game $\GameTuple$, we denote by $\StrategySet_\PlayerIndex^\diamond$ the set of distributions over $\StrategySet_\PlayerIndex$; this set is called agent $\PlayerIndex$'s \mydef{mixed strategy set}. We define $\smash{\MixedStratProfileSpace = \MixedStrategySet_1 \times \dots \times \MixedStrategySet_{\NumberOfPlayers}}$, and then, overloading notation, we write $\Utility(\StratProfile)$
to denote the expected utility of a mixed strategy profile $\smash{\StratProfile \in \MixedStratProfileSpace}$.

A solution to a normal-form game is a prediction of how strategic agents will play the game. One solution concept that has received a great deal of attention in the literature is Nash equilibrium~\cite{nash1950equilibrium}, a (pure or mixed) strategy profile at which each agent selects a utility-maximizing strategy, fixing all other agents' strategies. In this paper, we are concerned with $\epsilon$-Nash equilibrium, an approximation of Nash equilibrium that is amenable to statistical estimation. 

\begin{definition}[Regret]
Given a game $\GameTuple$, 
fix an agent $\PlayerIndex$
and a mixed strategy profile $\StratProfile \in \MixedStratProfileSpace$.
We define $\MixedStratProfileSpaceAlt_{\PlayerIndex, \StratProfile} \doteq \{ \StratProfileAlt \in \MixedStratProfileSpace \mid \StratProfileAlt_q = \StratProfile_q, \forall q \neq \PlayerIndex \}$.
In words, $\MixedStratProfileSpaceAlt_{\PlayerIndex, \StratProfile}$ is the set of all mixed strategy profiles in which the strategies of all agents $q \ne \PlayerIndex$ are fixed at $\StratProfile_q$.
Agent $\PlayerIndex$'s \mydef{regret} at $\StratProfile$ 
is defined as:
$\smash{\Regret^{\diamond}_\PlayerIndex (\GameTuple, \StratProfile) \doteq \sup_{\StratProfile' \in \MixedStratProfileSpaceAlt_{\PlayerIndex, \StratProfile}} \Utility_{\PlayerIndex} (\StratProfile') - \Utility_{\PlayerIndex} (\StratProfile)}$.
By restricting $\StratProfile$ and $\MixedStratProfileSpaceAlt_{\PlayerIndex, \StratProfile}$ to pure strategy profiles,
agent $\PlayerIndex$'s \mydef{pure regret} $\Regret_\PlayerIndex (\GameTuple, \StratProfile)$ can be defined similarly.
\end{definition}

Note that $\Regret^{\diamond}_\PlayerIndex (\GameTuple, \StratProfile), \Regret_\PlayerIndex (\GameTuple, \StratProfile) \ge 0$, since agent $\PlayerIndex$ can deviate to any strategy $\Strategy' \in \StrategySet_\PlayerIndex$, including $\StratProfile_\PlayerIndex$ itself.
A strategy profile $\StratProfile$ that has regret at most $\epsilon \ge 0$, for all $\PlayerIndex \in \SetOfPlayers$, is an $\epsilon$-Nash equilibrium:

\begin{definition}[$\epsilon$-Nash Equilibrium]
Given $\epsilon \ge 0$,
a mixed strategy profile $\StratProfile \in \MixedStratProfileSpace$ in a game $\GameTuple$ is an $\epsilon$-\mydef{Nash equilibrium} if, for all $\PlayerIndex \in \SetOfPlayers$, 
$\Regret^{\diamond}_\PlayerIndex (\GameTuple, \StratProfile) \le \epsilon$.
At a pure strategy $\epsilon$-Nash equilibrium $\StratProfile \in \StratProfileSpace$, for all $\PlayerIndex \in \SetOfPlayers$,
$\Regret_\PlayerIndex (\GameTuple, \StratProfile) \le \epsilon$.
We denote by $\Nash^\diamond_\epsilon (\GameTuple)$ %
the set of mixed $\epsilon$-Nash equilibria, and by $\Nash_\epsilon (\GameTuple)$, 
the set of pure $\epsilon$-Nash equilibria. Note that $\Nash_\epsilon(\GameTuple) \subseteq \Nash_\epsilon^\diamond(\GameTuple)$.
\end{definition}

\subsection{Approximating Equilibria}

In this section, we show that equilibria can be approximated with bounded error, given only a uniform approximation. Specifically, our main theorem establishes perfect recall, in the sense that the approximate game contains all true positives: i.e., all (exact) equilibria of the original game.
It also establishes approximately perfect precision, in the sense that all false positives in the approximate game are approximate equilibria in the original game.

We define the \emph{$\ell_{\infty}$-norm} between two compatible games, with the same agents sets $\SetOfPlayers$ and strategy profile spaces $\StratProfileSpace$, and with utility functions $\Utility$, $\Utility'$, respectively, as follows:%
\footnote{We use $\sup$ in this definition, rather than $\max$, because the space of mixed strategy profiles is infinite.}
\[
\norm{\GameTuple - \GameTuple'}_{\!\infty} \doteq \norm{\Utility(\cdot) - \Utility' (\cdot)}_{\!\infty} \doteq \smallsup{\PlayerIndex \in \SetOfPlayers, \StratProfile \in \StratProfileSpace} \lvert{\Utility_{\PlayerIndex} (\StratProfile) - \Utility'_{\PlayerIndex} (\StratProfile)}\rvert \enspace.
\]
While the $\ell_{\!\infty}$-norm as defined applies only to pure normal-form games, it is in fact sufficient to use this metric even to show that the utilities of mixed strategy profiles approximate one another.  We formalize this claim in the following lemma.

\begin{lemma}[Approximations in Mixed Strategies]
\label{lemma:mixed-approximation}
If $\GameTuple, \GameTuple'$ differ only in their utility functions, $\Utility$ and $\Utility'$, then
$\sup_{\PlayerIndex \in \SetOfPlayers, \StratProfile \in \MixedStratProfileSpace} \lvert{\Utility_\PlayerIndex(\StratProfile) - \Utility'_\PlayerIndex(\StratProfile)}\rvert = \norm{\GameTuple - \GameTuple'}_{\!\infty}$.
\end{lemma}

\begin{proof}
For any agent $\PlayerIndex$ and mixed strategy profile $\MixedStratProfileAlt \in \MixedStratProfileSpace$, $\Utility_\PlayerIndex (\MixedStratProfileAlt) = \sum_{\StratProfile \in \StratProfileSpace} \MixedStratProfileAlt (\StratProfile) \Utility_\PlayerIndex (\StratProfile)$, 
where $\MixedStratProfileAlt (\StratProfile) = \prod_{\PlayerIndex' \in \SetOfPlayers} \MixedStratProfileAlt_{\PlayerIndex'} (\StratProfile_{\PlayerIndex'})$.
So, $\Utility_\PlayerIndex (\MixedStratProfileAlt) - \Utility'_\PlayerIndex (\MixedStratProfileAlt) = 
\sum_{\StratProfile \in \StratProfileSpace} \MixedStratProfileAlt (\StratProfile) (\Utility_\PlayerIndex (\StratProfile) - \Utility'_\PlayerIndex (\StratProfile)) \leq
\sup_{\StratProfile \in \StratProfileSpace} \sabs{\Utility_\PlayerIndex (\StratProfile) - \Utility'_\PlayerIndex (\StratProfile)}$, by H{\"o}lder's inequality.
Hence,
$\sup_{\MixedStratProfileAlt \in \MixedStratProfileSpace} \sabs{\Utility_\PlayerIndex (\MixedStratProfileAlt) - \Utility'_\PlayerIndex (\MixedStratProfileAlt)} \leq
\sup_{\StratProfile \in \StratProfileSpace} \sabs{\Utility_\PlayerIndex (\StratProfile) - \Utility'_\PlayerIndex(\StratProfile)}$, from which it follows that
$\sup_{\PlayerIndex \in \SetOfPlayers, \MixedStratProfileAlt \in \MixedStratProfileSpace} \sabs{\Utility_\PlayerIndex (\MixedStratProfileAlt) - \Utility'_\PlayerIndex (\MixedStratProfileAlt)} \leq
\norm{\GameTuple - \GameTuple'}_{\!\infty}$.
Equality holds for any $\PlayerIndex$ and $\StratProfile$ that realize the supremum in $\norm{\GameTuple - \GameTuple'}_{\!\infty}$, as any such pure strategy profile is also mixed.
\end{proof}

\begin{definition}
$\Gamma'$ is said to be a \mydef{uniform $\epsilon$-approximation} of $\GameTuple$ when $\norm{\Gamma - \Gamma'}_{\!\infty} \leq \epsilon$.
\end{definition}

Uniform approximations are so-called
because the bound between utility deviations in $\GameTuple$ and $\GameTuple'$ holds \emph{uniformly\/} over \emph{all\/} players and strategy profiles.  
This notion is central to our work.

\begin{theorem}[Approximability of Equilibria]
\label{thm:equilibria-approximation-bounds}
If two normal-form games, $\GameTuple$ and $\GameTuple'$, are uniform approximations of one another,
then:

\begin{enumerate}
\item $\Nash(\GameTuple) \subseteq \Nash_{2\epsilon} (\GameTuple') \subseteq \Nash_{4\epsilon}(\GameTuple)$

\item $\Nash^\diamond (\GameTuple) \subseteq \Nash^{\diamond}_{2\epsilon} (\GameTuple') \subseteq \Nash^{\diamond}_{4\epsilon} (\GameTuple)$
\end{enumerate}
\end{theorem}

\begin{proof}
First note the following: if $\smash{A \subseteq B}$, then $\smash{C \cap A \subseteq C \cap B}$.
Hence, since any pure Nash equilibrium is also a mixed Nash equilibrium, taking $C$ to be the set of all pure strategy profiles, we need only show 2.
We do so by showing $\smash{\Nash^\diamond_\gamma (\GameTuple) \subseteq \Nash^\diamond_{2\epsilon + \gamma} (\GameTuple')}$, for $\gamma \ge 0$, which implies both containments, taking $\gamma = 0$ for the lesser, and $\gamma = 2\epsilon$ for the greater.

Suppose $\StratProfile \in \Nash^\diamond_\gamma (\GameTuple)$, for all $\PlayerIndex \in \SetOfPlayers$.
We will show that $\StratProfile$ is ($2\epsilon + \gamma$)-optimal in $\GameTuple$, for all $\PlayerIndex \in \SetOfPlayers$.
Fix an agent $\PlayerIndex$, and
define $\MixedStratProfileSpaceAlt_{\PlayerIndex, \StratProfile} \doteq \{ \MixedStratProfileAlt \in \MixedStratProfileSpace \mid \MixedStratProfileAlt_q = \StratProfile_q, \forall q \neq \PlayerIndex \}$.
In words, $\MixedStratProfileSpaceAlt_{\PlayerIndex, \StratProfile}$ is the set of all mixed strategy profiles in which the strategies of all agents $q \ne \PlayerIndex$ are fixed at $\StratProfile_q$.
Now take
$\StratProfile^* \in \argmax_{\MixedStratProfileAlt \in \MixedStratProfileSpaceAlt_{\PlayerIndex, \StratProfile}} \Utility_\PlayerIndex (\MixedStratProfileAlt)$ and
$\StratProfile'^* \in \argmax_{\MixedStratProfileAlt \in \MixedStratProfileSpaceAlt_{\PlayerIndex, \StratProfile}} \Utility'_\PlayerIndex (\MixedStratProfileAlt)$.
Then:
\begin{align*}
\Regret_\PlayerIndex (\GameTuple, \StratProfile)
&= \Utility'_\PlayerIndex({\StratProfile'}^*) - \Utility'_\PlayerIndex(\StratProfile) \\
&\leq (\Utility_\PlayerIndex({\StratProfile'}^*) + \epsilon) - (\Utility_\PlayerIndex(\StratProfile) - \epsilon) \\
&\leq (\Utility_\PlayerIndex(\StratProfile^*) + \epsilon) - (\Utility_\PlayerIndex(\StratProfile) - \epsilon) \\
 &\leq (\Utility_\PlayerIndex(\StratProfile^*) + \epsilon) - (\Utility_\PlayerIndex(\StratProfile^*) - \epsilon - \gamma) = 2\epsilon + \gamma
\end{align*}

\noindent
The first line follows by definition.
The second holds by Lemma~\ref{lemma:mixed-approximation} and the fact that $\Gamma'$ is a uniform $\epsilon$-approximation of $\Gamma$, the second as $\StratProfile^*$ is optimal for $\PlayerIndex$ in $\GameTuple$, and the third as $\StratProfile$ is a $\gamma$-Nash in $\GameTuple$.
\end{proof}

\section{Learning Framework}
\label{sec:learn}

In this section, we move on from approximating equilibria in games to learning them.
We present algorithms that learn so-called empirical games, which comprise estimates of the expected utilities of simulation-based games.
We further derive \mydef{uniform convergence bounds},
proving that our algorithms output empirical games that uniformly approximate their expected counterparts, with high probability.
Therefore, the equilibria of these empirical games approximate those of the corresponding simulation-based games, with high probability.

\subsection{Empirical Game Theory}

We start by developing a formalism for modeling simulation-based games, which we use to formalize empirical games, the main object of study in empirical game-theoretic analysis.

\begin{definition}[Conditional Normal-Form Game]
\label{def:conditional-game}
A \mydef{conditional normal-form game} $\ConditionalGame{\ConditionSpace} \doteq \langle \ConditionSpace, \SetOfPlayers, \{ \StrategySet_\PlayerIndex \mid \PlayerIndex \in \SetOfPlayers \}, \Utility(\cdot) \rangle$ consists of a set of conditions $\smash{\ConditionSpace}$, a set of agents $\smash{\SetOfPlayers}$, with pure strategy set $\smash{\StrategySet_\PlayerIndex}$ available to agent $\PlayerIndex$, and a vector-valued conditional utility function $\smash{\Utility : \StratProfileSpace \times \ConditionSpace \to \R^{\NumberOfPlayers}}$.
Given a condition $\smash{\ConditionValue \in \ConditionSpace}$, $\Utility(\cdot; \ConditionValue)$ yields a standard utility function of the form $\smash{\StratProfileSpace \to \R^{\NumberOfPlayers}}$. %
\end{definition}

\begin{definition}[Expected Normal-Form Game]
\label{def:expected-game}
Given a conditional normal-form game $\smash{\ConditionalGame{\ConditionSpace}}$ together with distribution $\ConditionDistribution$, we define the %
\mydef{expected utility function} $\smash{\Utility(\StratProfile; \ConditionDistribution) = \Expwrt{\ConditionValue \distributed \ConditionDistribution}{\Utility(\StratProfile; \ConditionValue)}}$, and the corresponding %
\mydef{expected normal-form game} as $\smash{\InducedGame{\ConditionDistribution} \doteq \langle \SetOfPlayers, \{ \StrategySet_\PlayerIndex \mid \PlayerIndex \in \SetOfPlayers \}, \Utility(\cdot; \ConditionDistribution) \rangle}$.
\end{definition}

Expected normal-form games serve as our mathematical model of simulation-based games.  They are sufficient not only to model arbitrary black-box games, but additionally games where the \emph{rules} are known but \emph{environmental conditions} are random (e.g., auctions where bidder valuations are random, or \emph{deterministic} war games where initial armies and terrain are random), as well as games with \emph{randomness}, where $\ConditionSpace$ is taken to be a \emph{PRNG seed} or \emph{entropy source}.

\begin{definition}[Empirical Normal-Form Game]
\label{def:empirical-game}
Given a conditional normal-form game $\smash{\ConditionalGame{\ConditionSpace}}$ together with a distribution $\ConditionDistribution$ from which we can draw samples $\smash{\Samples = (\SamplePoint_1, \ldots, \SamplePoint_\NumberOfSamples) \distributed \ConditionDistribution^\NumberOfSamples}$, we define the %
\mydef{empirical utility function} $\smash{\hat{\Utility}(\StratProfile; \Samples) \doteq \frac{1}{\NumberOfSamples}\sum_{\SampleIndex=1}^\NumberOfSamples \Utility(\StratProfile; \SamplePoint_\SampleIndex)}$, and the corresponding 
\mydef{empirical normal-form game} as $\smash{\EmpiricalGame{\Samples} \doteq \langle \SetOfPlayers, \{ \StrategySet_\PlayerIndex \mid \PlayerIndex \in \SetOfPlayers \}, \hat{\Utility}(\cdot ; \Samples)\rangle}$.
\end{definition}

\begin{observation}[Learnability] 
Consider a conditional normal-form game $\smash{\ConditionalGame{\ConditionSpace}}$ together with a distribution $\ConditionDistribution$ and $\Samples \sim \ConditionDistribution^\NumberOfSamples$, and the corresponding expected and empirical games, namely $\smash{\InducedGame{\ConditionDistribution}}$ and $\smash{\EmpiricalGame{\Samples}}$.
If, for some $\epsilon, \delta > 0$, $\displaystyle \mathbb{P}_{\Samples \sim \ConditionDistribution^{\NumberOfSamples}} \left( \smash{\norm{\InducedGame{\ConditionDistribution} - \EmpiricalGame{\Samples}}_{\!\infty} \le \epsilon} \right) \ge 1 - \delta$, then the
equilibria of $\smash{\InducedGame{\ConditionDistribution}}$ are learnable: i.e., the equilibria of $\smash{\EmpiricalGame{\Samples}}$ estimate the equilibria of $\smash{\EmpiricalGame{\ConditionDistribution}}$ up to additive error $\epsilon$ with probability at least $1 - \delta$.
\label{obs:learnability}
\end{observation}

\begin{proof}
Taking $\GameTuple$ to be $\InducedGame{\ConditionDistribution}$ and $\GameTuple'$ to be $\EmpiricalGame{\Samples}$, the dual conclusions of
\Cref{thm:equilibria-approximation-bounds} hold with probability at least $1 - \delta$.
\end{proof}

Our present goal, then, is to ``uniformly learn'' empirical games (i.e., obtain uniform convergence guarantees) from finitely many samples.
As per Observation~\ref{obs:learnability}, we can then apply the machinery of \Cref{thm:equilibria-approximation-bounds} to infer guarantees on the equilibria of simulation-based games.
As already noted, this learning problem is non-trivial because it involves multiple comparisons.
We describe two potential solutions, both of which are intended to control the FWER.
The first is a classical method: it applies a Bonferroni correction to multiple per-parameter confidence intervals derived via Hoeffding's inequality; the second uses Rademacher averages.
Both approaches yield
bounds on the rate at which \emph{all\/} utility estimates converge to their expectations.

\subsection{Hoeffding's Inequality}
\label{sec:concentration:hoeffding}

Hoeffding's inequality for sums of independent bounded random variables can be used to obtain tail bounds on the probabilities that empirical utilities differ greatly from their expectations.
In the next theorem, we use this inequality to estimate a single utility value, and then apply a union bound to estimate all utility values simultaneously.
Note that this theorem applies only to finite games: i.e., games for which the index set $\Indices \subseteq \SetOfPlayers \times \StratProfileSpace$ is finite.
Given a finite game, we state all bounds for an arbitrary index set $\Indices$; then, taking $\Indices = \SetOfPlayers \times \StratProfileSpace$,
our bounds imply bounds on $\smash{\norm{\Utility(\cdot; \ConditionDistribution) - \hat{\Utility}(\cdot; \Samples)}_{\!\infty}}$.

\begin{theorem}[Finite-Sample Bounds for expected Normal-Form Games via Hoeffding's Inequality]
Consider finite, conditional normal-form game $\ConditionalGame{\ConditionSpace}$ together with distribution $\ConditionDistribution$ and index set $\smash{\Indices \subseteq \SetOfPlayers \times \StratProfileSpace}$ such that for all $\smash{\ConditionValue \in \ConditionSpace}$ and $\smash{(\PlayerIndex, \StratProfile) \in \Indices}$, it holds that $\smash{\Utility_{\PlayerIndex}(\StratProfile; \ConditionValue) \in [-{\nicefrac{\UtilityRange}{2}}, {\nicefrac{\UtilityRange}{2}}]}$,
where $c \in \mathbb{R}$.
Then, with probability at least $\smash{1 - \delta}$ over $\smash{\Samples \distributed \ConditionDistribution^\NumberOfSamples}$, 
we may bound the deviation between $\Utility(\cdot; \ConditionDistribution)$ and $\hat{\Utility}(\cdot; \Samples)$ for a single $\smash{(\PlayerIndex, \StratProfile) \in \Indices}$, and for all $\smash{(\PlayerIndex, \StratProfile) \in \Indices}$, respectively, as:

\begin{enumerate}
\item $\displaystyle %
\left\lvert \Utility_{\PlayerIndex}(\StratProfile; \ConditionDistribution) - \hat{\Utility}_{\PlayerIndex}(\StratProfile; \Samples) \right\rvert \le \UtilityRange\sqrt{\frac{\ln \left(  \nicefrac{2}{\delta} \right)}{2\NumberOfSamples}}$ %

\item $\displaystyle %
\sup_{(\PlayerIndex, \StratProfile) \in \Indices} \left\lvert \Utility_{\PlayerIndex}(\StratProfile; \ConditionDistribution) - \hat{\Utility}_{\PlayerIndex}(\StratProfile; \Samples) \right\rvert \le \UtilityRange\sqrt{\frac{\ln \left( \nicefrac{2 \abs{\Indices}}{\delta} \right)}{2\NumberOfSamples}}$ %

\end{enumerate}
\label{thm:Hoeffding}
\end{theorem}

At a high-level, the first claim follows immediately from Hoeffding's inequality, and the second, via a union bound.
A formal proof of this theorem appears in Appendix~\ref{sec:proofs}.

By assuming independence among agents' utilities, and then applying a
\v{S}id\'{a}k rather than a Bonferroni correction, a slightly tighter bound was derived previously~\cite{tuyls2018generalised}.%
\footnote{Specifically, we use $\mathbb{P}(A \vee B) \leq \mathbb{P}(A) + \mathbb{P}(B)$, whereas previous work~\cite{tuyls2018generalised} assumes independence and uses $\mathbb{P}(A \vee B) = 1 - (1 - \mathbb{P}(A))(1 - \mathbb{P}(B))$.}
Note, however, that utilities in simulation-based games could very well exhibit dependencies. 
For example, imagine a simulation-based game where agents' utilities depend on the weather: e.g., on snowy winter days, all utilities are high, while on rainy winter days, they are all low.
To learn the utilities in this game, one could first sample the weather, and then sample utilities conditioned on the weather.
Our bound holds in this case, so long as (only) the weather samples are independent.  That utilities are independent, assuming an arbitrary black-box simulator, is a needlessly strong assumption.

\subsection{Rademacher Averages}
\label{section:concentration:rademacher}

Because Hoeffding's inequality assumes only bounded noise, it is often a loose bound, particularly for random variables that are tightly concentrated in a relatively small region of their domain.  
Theorem~\ref{thm:Hoeffding} and Section 4.2 in~\cite{tuyls2018generalised} both use Hoeffding's inequality to first bound individual utilities, and then combine those individual bounds using a statistical correction procedure (i.e., Bonferroni and \v{S}id\'{a}k, respectively).  Unfortunately, these corrections amplify the errors: e.g., Theorem~\ref{thm:Hoeffding} requires an $\smash{\mathcal{O}(\sqrt{\ln \abs{\Indices}})}$ factor increase in the widths of the confidence intervals.

In contrast, Rademacher averages---which
intuitively resemble permutation tests---bound
the deviation between (potentially infinitely) many utilities and their expectations directly.
For very large games, a Rademacher-based approach can yield bounds that are significantly tighter than classical methods.
For small and medium-sized games, however, they are often looser.  Nonetheless, we show that under natural conditions, they are never looser by more than a constant factor.

There are two equivalent definitions of Rademacher averages, one with families of \emph{sets}, which is often more natural for unsupervised learning problems, and one with families of \emph{functions}, which is usually more convenient for supervised learning.  We present yet a third, isomorphic, definition, which facilitates the analysis of expected normal-form games.

\begin{definition}[Rademacher Averages of Expected Normal-Form Games]
Consider a conditional normal-form game $\ConditionalGame{\ConditionSpace}$ together with distribution $\ConditionDistribution$ and take index set $\Indices \subseteq \SetOfPlayers \times \StratProfileSpace$.  Suppose $\smash{\Samples \doteq (\SamplePoint_1, \ldots, \SamplePoint_{\NumberOfSamples}) \distributed \ConditionDistribution^\NumberOfSamples}$, and take $\smash{\vsigma \distributed \Rademacher^\NumberOfSamples}$; the Rademacher distribution is uniform on $\pm 1$.  We define:

\begin{enumerate}

\item \textbf{1-Draw Empirical Rademacher Average} (1-ERA):
$\DERC{\NumberOfSamples}{\GameTuple, \Indices}{\Samples}{\vsigma} \doteq \sup\limits_{(\PlayerIndex, \StratProfile) \in \Indices} \Big\lvert \frac{1}{\NumberOfSamples}\sum_{\SampleIndex = 1}^\NumberOfSamples \vsigma_\SampleIndex \Utility_{\PlayerIndex}(\StratProfile ; \SamplePoint_\SampleIndex) \Big\rvert$

\item \textbf{Rademacher Average} (RA):
$\RC{\NumberOfSamples}{\GameTuple,\Indices}{\ConditionDistribution} \doteq \Expwrt{\Samples,\vsigma}{\DERC{\NumberOfSamples}{\GameTuple, \Indices}{\Samples}{\vsigma}}$
\end{enumerate}

\end{definition}

Observe that the 1-ERA can be computed exactly from one draw of $\Samples$ and $\vsigma$, whereas computing the RA involves taking an expectation over these random variables.  The RA is used to obtain tail bounds on the quality of estimates, but since the ERA is tightly concentrated about the RA, the ERA can also be used for this purpose.

\begin{theorem}[Finite-Sample Rademacher Bounds for Expected Normal-Form Games]
Consider conditional normal-form game $\ConditionalGame{\ConditionSpace}$ together with distribution $\ConditionDistribution$ and index set $\smash{\Indices \subseteq \SetOfPlayers \times \StratProfileSpace}$ such that for all $\smash{\ConditionValue \in \ConditionSpace}$ and $\smash{(\PlayerIndex, \StratProfile) \in \Indices}$, it holds that $\smash{\Utility_{\PlayerIndex}(\StratProfile; \ConditionValue) \in [-{\nicefrac{\UtilityRange}{2}}, {\nicefrac{\UtilityRange}{2}}]}$, 
where $c \in \mathbb{R}$.
Then, with probability at least $\smash{1 - \delta}$ over $\smash{\Samples \distributed \ConditionDistribution^\NumberOfSamples}$, we may bound the deviation between $\Utility(\cdot; \ConditionDistribution)$ and $\hat{\Utility}(\cdot; \Samples)$ over all $\smash{(\PlayerIndex, \StratProfile) \in \Indices}$ with the 1-ERA and the RA, respectively, as:
\begin{enumerate}
\item
$\displaystyle \sup_{(\PlayerIndex, \StratProfile) \in \Indices} \abs{\Utility_{\PlayerIndex} (\StratProfile; \ConditionDistribution) - \hat{\Utility}_{\PlayerIndex} (\StratProfile; \Samples)} \leq 2\DERC{\NumberOfSamples}{\GameTuple, \Indices}{\Samples}{\vsigma} + 3\UtilityRange\sqrt{\frac{\ln\left(\nicefrac{1}{\delta}\right)}{2\NumberOfSamples}}$

\item
$\displaystyle
\sup_{(\PlayerIndex, \StratProfile) \in \Indices} \abs{\Utility_{\PlayerIndex} (\StratProfile; \ConditionDistribution) - \hat{\Utility}_{\PlayerIndex} (\StratProfile; \Samples)} \leq 2\RC{\NumberOfSamples}{\GameTuple, \Indices}{\ConditionDistribution} + \UtilityRange\sqrt{\frac{\ln\left(\nicefrac{1}{\delta}\right)}{2\NumberOfSamples}} \leq \UtilityRange \sqrt{\frac{\ln(\abs{\Indices})}{2\NumberOfSamples}} + \UtilityRange\sqrt{\frac{\ln\left(\nicefrac{1}{\delta}\right)}{2\NumberOfSamples}}$
\end{enumerate}
\label{thm:Rademacher}
\end{theorem}

A proof of~\Cref{thm:Rademacher} appears in~\Cref{sec:proofs}.

Observe that the 1-ERA bound, Theorem~\ref{thm:Rademacher} (1), has no explicit dependence on the number of parameters $\abs{\Indices}$ being estimated; thus, unlike a correction-based approach, this bound can scale to arbitrarily large games (including infinite games) without necessarily incurring additive error $\Omega(\sqrt{\ln{\abs{\Indices}}})$.
Moreover, the upper bound on the RA, Theorem~\ref{thm:Rademacher} (2), shows that the Rademacher bounds are never asymptotically worse than those obtained via Theorem~\ref{thm:Hoeffding}.

On the other hand, the 1-ERA is a data-dependent quantity; thus, it does not yield strong \emph{a priori\/} bounds.  Still, the 1-ERA bound %
can be much tighter
than the RA upper bound, in which case the bounds of Theorem~\ref{thm:Rademacher} may outperform those of Theorem~\ref{thm:Hoeffding}.
For example, when utility values are \emph{correlated} with respect to $\ConditionDistribution$, then the supremum is over 
correlated variables, which is generally smaller than a supremum over 
uncorrelated variables---strongly correlated sets of variables behave like single variables in the supremum.  Even under anticorrelation, as in constant-sum games, we obtain this benefit, due to the absolute value in the 1-ERA definition.  Additionally, if some $\smash{(\PlayerIndex, \StratProfile) \in \Indices}$ experience more noise than others, the low-noise indices should have a small effect on the 1-ERA, as they are less likely than high-noise indices to realize the supremum.

In \Cref{sec:comparingBounds}, 
we design stylized games with particular properties for which we can prove that the 
bounds of \Cref{thm:Rademacher} are significantly tighter than those of \Cref{thm:Hoeffding}.  Crucially, this performance holds regardless of whether it is known \emph{a priori} that such properties hold.  Consequently, the Rademacher bounds are an attractive choice when we suspect a game is well-behaved, but have insufficient \emph{a priori\/} knowledge to prove it possesses any particular property that would imply strong statistical bounds, as is typical of the simulation-based games analyzed using EGTA.  

\begin{remark}
\label{rem:size}
The second bound in Theorem~\ref{thm:Rademacher} tells us that the Rademacher bounds are never asymptotically worse than those obtained via Theorem~\ref{thm:Hoeffding}.  Their small-sample performance, however, is inferior.  Indeed, we can derive the minimum game size in which the former may outperform the latter by setting $c\sqrt{\nicefrac{\lnp{2\abs{\Indices}/\delta}}{2m}} = 3c\sqrt{\nicefrac{\lnp{1/\delta}}{2m}}$, and solving for $\abs{\Indices}$.
Doing so yields $\abs{\Indices} = \nicefrac{1}{2\delta^8}$.  Taking $\delta = 0.1$, we obtain a minimum $\abs{\Indices}$ of $5 \times 10^7$.  Solving games of this magnitude is not entirely out of reach, even today, and with future increases in computing power, the asymptotic efficiency of the Rademacher bounds will only become more attractive.
\end{remark}

\begin{remark}
While estimating an arbitrary game requires estimating $\smash{\NumberOfPlayers\prod_{i=1}^{\NumberOfPlayers}|\StrategySet_\PlayerIndex|}$ parameters, many games are defined by far fewer parameters.
\mydef{Symmetric games}, for example, where all agents' strategy sets are identical, and utilities depend only on the strategy an agent chooses and the \emph{number\/} of other agents that choose each strategy (not the agents' identities), can be substantially smaller.
Learning such games is easier; as fewer parameters need to be estimated, we can obtain tighter statistical bounds from the same number of samples.  
\end{remark}

\begin{remark}
Rademacher and Hoeffding bounds are just two possible choices of concentration inequalities.  Both require bounded noise, but this is not an inherent limitation of our approach.
We could obtain similar results under varied noise assumptions; e.g., we could assume (unbounded) subgaussian or subexponential noise, and substitute the appropriate Chernoff bounds.
\end{remark}

\section{Learning Algorithms}
\label{sec:algos}

We are now ready to present our algorithms.
Specifically, we discuss two Monte-Carlo sampling-based algorithms that can be used to uniformly learn empirical games, and hence ensure that the equilibria of the games they are learning are accurately approximated with high probability.
Note that our algorithms apply only to finite games, as they require an enumeration of the index set $\Indices$.

A conditional normal form game $\ConditionalGame{\ConditionSpace}$, together with a black box from which we can sample distribution $\ConditionDistribution$, serves as our mathematical model of a black-box simulator from which the utilities of a simulation-based game can be sampled.  Given strategy profile $\StratProfile$, we
assume the simulator outputs a sample $\Utility_\PlayerIndex (\StratProfile, \ConditionValue)$, for \emph{all\/} agents $\PlayerIndex \in \SetOfPlayers$, after drawing a \emph{single\/} condition value $\ConditionValue \sim \ConditionDistribution$.

Our first algorithm, \mydef{global sampling} (\GS), is a straightforward application of Theorems~\ref{thm:Hoeffding} and~\ref{thm:Rademacher}. The second, \mydef{progressive sampling with pruning} (\PSP), iteratively prunes strategies, and thereby has the potential to expedite learning by obtaining tighter bounds than \GS, given the same number of samples.
We explore PSP's potential savings in our experiments (Section~\ref{sec:expts}).

\subsection{Global Sampling}
\label{subsec:GS}

Our first algorithm, \GS{} (\Cref{alg:gs}), samples all utilities of interest, given a sample size $\NumberOfSamples$ and a failure probability $\delta$, and returns the ensuing empirical game together with an $\hat{\epsilon}$ determined by either~\Cref{thm:Hoeffding} or~\Cref{thm:Rademacher} that guarantees an $\hat{\epsilon}$-uniform approximation.

More specifically,
\GS{} takes in a conditional game $\ConditionalGame{\ConditionSpace}$, a black box from which we can sample distribution $\ConditionDistribution$, an index set $\smash{\Indices \subseteq \SetOfPlayers \times \StratProfileSpace}$, a sample size $\NumberOfSamples$, a utility range $\UtilityRange$ such that utilities are required to lie in $[-\nicefrac{\UtilityRange}{2}, \nicefrac{\UtilityRange}{2}]$, and a bound type {\sc Bound}, and then draws $\NumberOfSamples$ samples to produce an empirical game $\hat{\GameTuple}_{\Samples}$, represented by $\tilde{\Utility} (\cdot)$, as well as an additive error $\hat{\epsilon}$,
with the following guarantee:

\begin{algorithm}[htbp]
\begin{algorithmic}[1]
\Procedure{GS}{$\ConditionalGame{\ConditionSpace}, \ConditionDistribution,  \Indices, \NumberOfSamples, \delta, \UtilityRange, \textsc{Bound}$} $\to (\tilde{\Utility}, \hat{\epsilon})$

\State \Input Conditional game $\ConditionalGame{\ConditionSpace}$, black box 
from which we can sample distribution $\ConditionDistribution$,
index set $\Indices$, sample size $\NumberOfSamples$, failure probability $\delta$, utility range $\UtilityRange$, and bound type \textsc{Bound}.

\State \Output Empirical utilities $\tilde{\Utility}, \forall (\PlayerIndex, \StratProfile) \in \Indices$, and additive error $\hat{\epsilon}$.

\State $\Samples \sim \ConditionDistribution^\NumberOfSamples$ \Comment{Draw $\NumberOfSamples$ samples from distribution $\ConditionDistribution$}

\State 
$\tilde{\Utility}_{\PlayerIndex} (\StratProfile) \gets
\hat{\Utility}_{\PlayerIndex}(\cdot; \Samples), \forall (\PlayerIndex, \StratProfile) \in \Indices$
\Comment{Compute empirical utilities}

\If{\textsc{Bound} = \textsc{1-ERA}}

    \State $\vsigma \gets \Rademacher^{\NumberOfSamples}$
    
    \State $r \gets \DERC{\NumberOfSamples}{\GameTuple, \Indices}{\Samples}{\vsigma}$
    \Comment{1-ERA of $\InducedGame{\ConditionDistribution}$}
    
    \State $\hat{\epsilon} \gets 2r + 3 \UtilityRange \sqrt{\nicefrac{\ln(\slfrac{1}{\delta})}{2\NumberOfSamples}}$
    \Comment{1-ERA bound}

\ElsIf{\textsc{Bound} = \textsc{Hoeffding}}

    \State $\hat{\epsilon} \gets \UtilityRange \sqrt{\nicefrac{\lnp{\slfrac{2\abs{\Indices}}{\delta}}}{2\NumberOfSamples}}$ \Comment{Hoeffding bound}

\EndIf

\State $\Return~(\tilde{\Utility}, \hat{\epsilon})$

\EndProcedure
\end{algorithmic}
\caption{Global Sampling}
\label{alg:gs}
\end{algorithm}

\begin{theorem}[Approximation Guarantees of Global Sampling]
\label{thm:gs-guarantee}
Consider conditional game $\smash{\ConditionalGame{\ConditionSpace}}$ together with distribution $\smash{\ConditionDistribution}$ and take index set $\smash{\Indices \subseteq \SetOfPlayers \times \StratProfileSpace}$ such that for all $\smash{\ConditionValue \in \ConditionSpace}$ and $\smash{(\PlayerIndex, \StratProfile) \in \Indices}$, $\smash{\Utility_{\PlayerIndex} (\StratProfile; \ConditionValue) \in [\nicefrac{\UtilityRange}{2}, \nicefrac{\UtilityRange}{2}]}$, for some $\smash{\UtilityRange \in \mathbb{R}}$.
If {\/} $\smash{\GS (\ConditionalGame{\ConditionSpace}, \ConditionDistribution, \Indices, \SetOfPlayers \times \StratProfileSpace, \NumberOfSamples, \delta, \UtilityRange, \textsc{Bound})}$ outputs the pair $\smash{(\tilde{\Utility}, \hat{\epsilon})}$, then with probability at least $\smash{1 - \delta}$, it holds that $\smash{\sup_{(\PlayerIndex, \StratProfile) \in \Indices} \abs{\Utility_{\PlayerIndex} (\StratProfile; \ConditionDistribution) - \tilde{\Utility}_{\PlayerIndex} (\StratProfile)} \leq \hat{\epsilon}}$.
\end{theorem}

\begin{proof}
The result follows from~\Cref{thm:Hoeffding} when $\textsc{Bound} = \textsc{Hoeffding}$ and~\Cref{thm:Rademacher} when $\textsc{Bound} = \textsc{1-ERA}$.
\end{proof}

\subsection{Progressive Sampling with Pruning}
\label{subsec:PSP}

Next, we present \PSP{} (\Cref{alg:psp_eqa}),
which, using \GS{} as a subroutine, draws progressively larger samples, refining the empirical game at each iteration, and stopping when the equilibria are approximated to the desired accuracy, or when the sampling budget is exhausted.
Although performance ultimately depends on a game's structure, \PSP{} can potentially learn equilibria using vastly fewer resources than \GS, when fewer data are necessary to learn to a desired degree of accuracy.

As the name suggests, \PSP{} is a pruning algorithm.  
The key idea is to prune (i.e., cease estimating the utilities of) strategy profiles that (w.h.p.) are provably not equilibria.
Recall that $\StratProfile \in \Nash_\epsilon (\GameTuple)$ iff $\Regret_\PlayerIndex (\GameTuple, \StratProfile) \le \epsilon$, for all $\PlayerIndex \in \SetOfPlayers$.
Thus, if there exists $\PlayerIndex \in \SetOfPlayers$ s.t.{} $\Regret_\PlayerIndex (\GameTuple, \StratProfile) > \epsilon$, then $\StratProfile \not\in \Nash_\epsilon (\GameTuple)$.
In the search for pure equilibria, such strategy profiles can be pruned.

For mixed equilibria, a strategy $\Strategy \in \StrategySet_\PlayerIndex$ is said to \mydef{$\epsilon$-dominate} another strategy $\Strategy' \in \StrategySet_\PlayerIndex$ if, for any pure strategy profile $\StratProfile$, taking $\StratProfile' = (\Strategy_1, \dots, \Strategy_{\PlayerIndex-1}, \Strategy', \Strategy_{\PlayerIndex+1}, \dots, \Strategy_{\NumberOfPlayers})$, it holds that $\Utility_\PlayerIndex(\StratProfile) + \epsilon \geq \Utility_\PlayerIndex(\StratProfile')$.
The \mydef{$\epsilon$-rationalizable} strategies $\Rationalizable_{\epsilon}(\GameTuple)$ are those that remain after iteratively removing all $\epsilon$-dominated strategies.  
Only strategies in $\Rationalizable_{\epsilon}(\GameTuple)$ can have nonzero weight in a mixed $\epsilon$-Nash equilibrium~\cite{gibbons1992game}; thus eliminating strategies not in $\Rationalizable_{\epsilon}(\GameTuple)$ is a natural pruning criterion.

If a strategy $\Strategy \in \StrategySet_\PlayerIndex$ is $\epsilon$-dominated by another strategy $\Strategy' \in \StrategySet_\PlayerIndex$, then $\PlayerIndex$ always regrets playing strategy $\Strategy$, regardless of other agents' strategies.
Consequently, the mixed pruning criterion is more conservative than the pure, so more pruning occurs when learning pure equilibria.  

Like \GS, \PSP{} takes in a conditional game $\ConditionalGame{\ConditionSpace}$, a black box from which we can sample distribution $\ConditionDistribution$, a utility range $\UtilityRange$, and a bound type {\sc Bound}.  
Instead of a single sample size, however, it takes in a \emph{sampling schedule} $\bm{M}$ in the form of a (possibly infinite) strictly increasing sequence of integers; and instead of a single failure probability, it takes in a \emph{failure probability schedule} $\bm{\delta}$, with each $\bm{\delta}_\TimeIndex$ in this sequence and their sum in $(0, 1)$.
These two schedules dictate the number of samples to draw and the failure probability to use at each iteration.
\PSP{} also takes in a boolean \textsc{Pure} that determines whether the output is pure (or mixed) equilibria, and an \emph{error threshold} $\epsilon$, which enables early termination as soon as equilibria of the desired sort are estimated to within the additive factor $\epsilon$.

\begin{algorithm}[htbp]
\algrenewcommand\algorithmicindent{1.0em}
\begin{algorithmic}[1]
\Procedure{PSP}{$\ConditionalGame{\ConditionSpace}, \ConditionDistribution, \SamplingSchedule, \bm{\delta}, \UtilityRange, \textsc{Bound}, \textsc{Pure}, \epsilon$} $\to ((\tilde{\Utility}, \tilde{\bm{\epsilon}}), (\hat{E}, \hat{\epsilon}), \hat{\delta})$
 
\State \Input 
Conditional game $\ConditionalGame{\ConditionSpace}$, black box 
from which we can sample distribution $\ConditionDistribution$, sampling schedule $\SamplingSchedule$, failure probability schedule $\bm{\delta}$, utility range $\UtilityRange$, bound type \textsc{Bound}, equilibrium type \textsc{Pure}, error threshold $\epsilon$.

\State \Output Empirical utilities $\tilde{\Utility}$, $\forall (\PlayerIndex, \StratProfile) \in \SetOfPlayers \times \StratProfileSpace$, utility error $\tilde{\bm{\epsilon}}$, empirical equilibria $\hat{E}$, equilibria error $\hat{\epsilon}$, failure probability $\hat{\delta}$.

\State $\Indices \gets \SetOfPlayers \times \StratProfileSpace$ \Comment{Initialize indices} \label{alg:psp:init-indices}

\State $(\tilde{\Utility}_{\PlayerIndex}(\StratProfile), \tilde{\bm{\epsilon}}_{\PlayerIndex}(\StratProfile)) \gets (0, \nicefrac{c}{2}), \forall (\PlayerIndex, \StratProfile) \in \Indices$
\Comment{Initialize outputs} \label{alg:psp:init-outputs}

\For{$\TimeIndex \in 1, \dots, \abs{\SamplingSchedule}$} %

    \State $(\tilde{\Utility}, \hat{\epsilon}) \gets \mathrm{\GS}(\ConditionalGame{\ConditionSpace}, \ConditionDistribution, \Indices, \SamplingSchedule_\TimeIndex, \bm{\delta}_\TimeIndex, \UtilityRange, \textsc{Bound})$ \label{alg:psp:utility-gs} \Comment{Improve utility estimates}

    \State $\tilde{\bm{\epsilon}}_{\PlayerIndex}(\StratProfile) \gets \hat{\epsilon}, \forall (\PlayerIndex,\StratProfile) \in \Indices$ \label{alg:psp:epsilon-gs} \Comment{Update confidence intervals}

    \If{$\hat{\epsilon} \leq \epsilon$ or $\TimeIndex = \abs{\SamplingSchedule}$} \Comment{Termination condition}
     \label{alg:psp:termination}

        \State $\hat{E} \gets \left\{ \begin{array}{rcl} \textsc{Pure} & : & \Nash_{2\hat{\epsilon}}(\tilde{\Utility}) \\ \neg\textsc{Pure} & : & \Nash^{\diamond}_{2\hat{\epsilon}}(\tilde{\Utility}) \\ \end{array} \right.$ %

    	\State $\Return~((\tilde{\Utility}, \tilde{\bm{\epsilon}}), (\hat{E}, \hat{\epsilon}), \sum_{i=1}^\TimeIndex \bm{\delta}_i)$
    	\label{alg:psp:return}
    \EndIf

    \State $\Indices \gets \hspace{-0.05cm} \left\{ \hspace{-0.15cm} \begin{array}{rcl}
    \textsc{Pure} & \hspace{-0.25cm} : \hspace{-0.3cm} & \{(\PlayerIndex, \StratProfile) \in \Indices \mid \Regret_\PlayerIndex (\tilde{\Utility}, \StratProfile) \leq 2\hat{\epsilon} \} \\
    \neg\textsc{Pure} & \hspace{-0.25cm} : \hspace{-0.3cm} & \{(\PlayerIndex, \StratProfile) \in \Indices \mid \StratProfile_q \in \Rationalizable_{2\hat{\epsilon}}(\tilde{\Utility}), \forall q \in \SetOfPlayers \}
    \end{array} \right.$
     \label{alg:psp:pruning}
    \Comment{Prune}

\EndFor

\EndProcedure
\end{algorithmic}

\caption{Progressive Sampling with Pruning}
\label{alg:psp_eqa}
\end{algorithm}

In the \PSP{} pseudocode, and in the following theorem, we overload the $\Regret$ and $\Nash$ operators (both pure and mixed) to depend on a utility function, rather than a game.

\begin{theorem}[Approximation Guarantees of Progressive Sampling with Pruning]
\label{thm:psp-guarantee}
Consider conditional game $\ConditionalGame{\ConditionSpace}$ together with distribution $\ConditionDistribution$ such that for all $\ConditionValue \in \ConditionSpace$ and $(\PlayerIndex, \StratProfile) \in \SetOfPlayers \times \StratProfileSpace$, $\Utility_{\PlayerIndex}(\StratProfile ; \ConditionValue) \in [-\nicefrac{\UtilityRange}{2}, \nicefrac{\UtilityRange}{2}]$, for some $\smash{\UtilityRange \in \mathbb{R}}$.  If {\/} $\PSP (\ConditionalGame{\ConditionSpace}, \ConditionDistribution, \SamplingSchedule, \bm{\delta}, \UtilityRange, \textsc{Bound}, \textsc{Pure}, \epsilon)$ outputs $\bigl((\tilde{\Utility}, \tilde{\bm{\epsilon}}), (\hat{E}, \hat{\epsilon}), \hat{\delta}\bigr)$, it holds that:

\begin{enumerate}
\item \label{thm:psp-guarantee:delta} $\hat{\delta} \leq \sum_{\bm{\delta}_\TimeIndex \in \bm{\delta}} \bm{\delta}_\TimeIndex, \hat{\delta} \in (0, 1)$.

\item \label{thm:psp-guarantee:epsilon} If $\lim_{\TimeIndex \to \infty} \ln(\nicefrac{1}{\bm{\delta}_\TimeIndex}) / \SamplingSchedule_\TimeIndex = 0$, then $\hat{\epsilon} < \epsilon$.
\end{enumerate}

\noindent
Furthermore, if {\/} \PSP{} terminates, then with probability at least $1 - \hat{\delta}$, the following hold simultaneously:

\begin{enumerate}
\setcounter{enumi}{2}

\item \label{thm:psp-guarantee:utility}
$\abs{\Utility_{\PlayerIndex}(\StratProfile; \ConditionDistribution) - \tilde{\Utility}_{\PlayerIndex}(\StratProfile)} \leq \tilde{\bm{\epsilon}}_{\PlayerIndex}(\StratProfile)$, for all $(\PlayerIndex, \StratProfile) \in \SetOfPlayers \times \StratProfileSpace$.

\item \label{thm:psp-guarantee:pure-nash} If \textsc{Pure}, then $\Nash(\Utility) \subseteq \Nash_{2\hat{\epsilon}}(\tilde{\Utility}) \subseteq \Nash_{4\hat{\epsilon}}(\Utility)$.

\item \label{thm:psp-guarantee:mixed-nash} If $\neg$\textsc{Pure}, then $\Nash^{\diamond}(\Utility) \subseteq \Nash^{\smash{\diamond}}_{\smash{2\hat{\epsilon}}}(\tilde{\Utility}) \subseteq \Nash^{\smash{\diamond}}_{\smash{4\hat{\epsilon}}}(\Utility)$.

\end{enumerate}
\end{theorem}

\begin{proof}

To see 1, note that $\hat{\delta}$ is computed on line \ref{alg:psp:return} as a partial sum of $\bm{\delta}$, each addend and the sum of which are all by assumption on $(0, 1)$; thus the result holds.

To see 2, note that if $\smash{\lim_{\TimeIndex \to \infty} \ln(\nicefrac{1}{\bm{\delta}_\TimeIndex}) / \SamplingSchedule_\TimeIndex = 0}$, then both the Hoeffding and Rademacher bounds employed by \GS{} tend to 0, as both decay asymptotically (in expectation) as $\smash{\mathcal{O}( \sqrt{\nicefrac{\ln(\slfrac{1}{\delta_\TimeIndex})}{\bm{M}_\TimeIndex}})}$ (see Theorems~\ref{thm:Hoeffding} and \ref{thm:Rademacher}). For infinite sampling schedules, the termination condition of line \ref{alg:psp:termination} ($\hat{\epsilon} \leq \epsilon$) is eventually met, as $\hat{\epsilon}$ is the output of \GS, and thus 2 holds.

To establish 3, we show the following:
assuming termination occurs at timestep $n$, with probability at least $1 - \hat{\delta}$, at every $\TimeIndex$ in $\{1, \dots, n\}$, it holds that $\smash{\sup_{(\PlayerIndex, \StratProfile) \in \SetOfPlayers \times \StratProfileSpace} \lvert{\Utility_\PlayerIndex(\StratProfile; \ConditionDistribution) - \tilde{\Utility}_{\PlayerIndex}(\StratProfile)}\rvert \leq \tilde{\bm{\epsilon}}_\PlayerIndex(\StratProfile)}$.
This property follows from the \GS{} guarantees of \Cref{thm:gs-guarantee}, as at each timestep $\TimeIndex$, the guarantee holds with probability at least $1 - \bm{\delta}_\TimeIndex$; thus by a union bound, the guarantees hold simultaneously at all time steps with probability at least $1 - \sum_{i=1}^n \bm{\delta}_\TimeIndex = 1 - \hat{\delta}$.
That the GS guarantees hold for unpruned indices should be clear; for pruned indices, since only error bounds for indices updated on line \ref{alg:psp:utility-gs} are tightened on line \ref{alg:psp:epsilon-gs}, it holds by the \GS{} guarantees of previous iterations.

Without pruning, 4 and 5 would follow directly from 3 via \Cref{thm:equilibria-approximation-bounds}, but with pruning, the situation is a bit more involved.

To see 4, observe that at each time step, only indices ($\PlayerIndex$, $\StratProfile$) such that $\smash{\Regret_\PlayerIndex \bigl( \tilde{\Utility} (\cdot), \StratProfile \bigr) > 2\hat{\epsilon}}$ are pruned (line \ref{alg:psp:pruning}), thus we may guarantee that with probability at least $\smash{1 - \hat{\delta}}$, $\smash{\Regret_\PlayerIndex \bigl( \Utility (\cdot; \ConditionDistribution), \StratProfile \bigr) > 0}$.  Increasing the accuracy of the estimates of these strategy profiles is thus not necessary, as they do not comprise pure equilibria (w.h.p.), and they will never be required to refute equilibria, as these will never be a best response for any agent from any strategy profile.

5 follows similarly, except that nonzero regret implies that a pure strategy profile is not a pure Nash equilibrium, but it does \emph{not\/} imply that it is \emph{not\/} part of any mixed Nash equilibrium.  Consequently, we use the more conservative pruning criterion of strategic dominance (a strategy is dominated if it is not rationalizable), requiring $2\epsilon$-dominance in $\tilde{\Utility}$, as this implies nonzero dominance
in $\Utility$.
\end{proof}

Finally, we propose two possible sampling and failure probability schedules for \PSP, $\bm{M}$ and $\bm{\delta}$, depending on whether the sampling budget is finite or infinite.  Given a finite sampling budget $m < \infty$, a neutral choice is to take $\bm{M}$ to be a doubling sequence such that $\sum_{M_i \in \bm{M}} M_i \leq m$, with $\bm{M}_1$ sufficiently large so as to possibly permit pruning after the first iteration (iterations that neither prune nor achieve $\epsilon$-accuracy are effectively wasted), and to take $\bm{\delta}_\TimeIndex = \nicefrac{\delta}{\abs{\bm{M}}}$, where $\delta$ is some maximum tolerable failure probability.  This strategy %
may fail to produce the desired $\epsilon$-approximation, as it may exhaust the sampling budget first.  
If we want to guarantee a particular $\epsilon$-$\delta$-approximation, then we can take $\bm{M}$ to be an infinite doubling sequence, and $\bm{\delta}$ to be a geometrically decreasing sequence such that $\sum_{\TimeIndex=1}^\infty \bm{\delta}_\TimeIndex = \delta$, for which the conditions of Theorem~\ref{thm:psp-guarantee} (\ref{thm:psp-guarantee:epsilon}) hold. 

\section{Experiments}
\label{sec:expts}

In this section, we describe the performance of our learning algorithms, experimentally.
We demonstrate that they consistently outperform their theoretical guarantees, in multiple experiments with random games, under various noise conditions.
Moreover, we show that in cases of interest, \PSP{} consistently requires fewer data than \GS.

We start by describing the games we designed for these experiments; we then present our findings. We use the notation $U[a, b]$, where $a < b$, to denote the uniform distribution, both in the discrete case ($a, b \in \N$)
and in the continuous case ($a, b \in \R$). Unless otherwise noted (i.e., in the empirical success rate experiments), $\delta$ is fixed at $0.1$ throughout. %

\subsection{Experimental Setup}
\label{subsec:experimentalSetup}

\paragraph{Uniform Random Games.} 
For $k \in \mathbb{Z}_+, u_0 \in \mathbb{R}$, a game $\GameTuple$ drawn from a uniform random game distribution $\RandomGame(\NumberOfPlayers, k, u_0)$ has $\NumberOfPlayers$ agents where each agent $\PlayerIndex$ has $k$ pure strategies. The utilities for each agent $\PlayerIndex$ in each profile $\StratProfile$ are drawn i.i.d.\ as $\smash{\Utility_\PlayerIndex (\StratProfile) \sim U(\nicefrac{-u_0}{2}, \nicefrac{u_0}{2})}$, where $u_0 > 0$ is a parameter controlling the magnitude of the utilities. We set $u_0 = 10$.

\paragraph{Finite Congestion Games.} 
Congestion games are a class of games known to exhibit pure strategy Nash equilibria~\cite{rosenthal1973class}. We use this class of games in our experiments for precisely this reason, so that we can more easily identify equilibria.

A tuple $\smash{\CongestionGame = (\SetOfPlayers, \SetOfFacilities, \{ \CongestionSetOfStrategies_\PlayerIndex \mid \PlayerIndex \in \SetOfPlayers \}, \{ \FacilityCostFunction_\FacilityIndex \mid \FacilityIndex \in \SetOfFacilities \})}$ is a \mydef{congestion game}~\cite{christodoulou2005price}, where $\smash{\SetOfPlayers = \{1, \ldots, \NumberOfPlayers\}}$ is a set of agents and $\smash{\SetOfFacilities = \{1, \ldots, \NumberOfFacilities\}}$ is a set of facilities. %
A strategy $\smash{\Strategy_\PlayerIndex \in \CongestionSetOfStrategies_\PlayerIndex} \subseteq 2^\SetOfFacilities$ is a set of facilities, and $\smash{\FacilityCostFunction_\FacilityIndex}$ is a (universal: i.e., non-agent specific) cost function associated with facility $\FacilityIndex$. 

In a finite congestion game, each agent prefers to select, among all their available strategies, one that minimizes their cost. Given a profile of strategies $\StratProfile$, agent $\PlayerIndex$'s cost is defined as $\smash{\PlayerCost_\PlayerIndex(\StratProfile) = \sum_{\FacilityIndex\in\Strategy_\PlayerIndex} \FacilityCostFunction_\FacilityIndex(n_\FacilityIndex(\StratProfile))}$, where $n_\FacilityIndex(\StratProfile)$ is the number of agents who select facility $\FacilityIndex$ in $\StratProfile$. We experiment with simple cost functions of the form $\FacilityCostFunction_\FacilityIndex(n) = n$.

For $\NumberOfPlayers \in \mathbb{Z}_+, \NumberOfFacilities \in \mathbb{Z}_+, k \in \mathbb{Z}_+$ such that $k \le 2^{\NumberOfFacilities} -1$,%
\footnote{We constrain $k \le 2^{\NumberOfFacilities} -1$, not $k \le 2^{\NumberOfFacilities}$, so that no agent's strategy set is empty.} 
a finite congestion game $\CongestionGame$ drawn from a random congestion distribution $\smash{\RandomCongestionGame(\NumberOfPlayers, \NumberOfFacilities, k)}$ is played by $\NumberOfPlayers$ agents on $\NumberOfFacilities$ facilities.
Each agent $\PlayerIndex$ is assigned a random number $\smash{\abs{\StrategySet_\PlayerIndex} \sim U[1, k]}$ of pure strategies. Each pure strategy $\Strategy_{\PlayerIndex, j} \in \StrategySet_\PlayerIndex$ is constructed as follows. The probability that facility $\FacilityIndex \in \{1, \ldots, \NumberOfFacilities\}$ is included in $\Strategy_{\PlayerIndex, j}$ is given by the power law $\smash{\mathbb{P} [\FacilityIndex \in \Strategy_{\PlayerIndex, j}] = \alpha^{-\FacilityIndex}}$, where $\alpha \in [0,1]$. We report results for $\alpha = 0.1$. This distribution is a simple but useful model of situations where agents' preferences are clustered around a few preferred facilities. Power law models abound in the literature; they are used to model distributions of population in cities, of incomes, and of species among genera, among many other applications~\cite{simon1955class, mitzenmacher2004brief}.

\paragraph{Noise simulation.} To test our algorithms, we query a black-box simulator to obtain noisy samples of the agents' utilities. Given a game $\GameTuple$, a sample of
agent $\PlayerIndex$'s utility at strategy profile $\StratProfile$ is given by $\smash{\Utility_\PlayerIndex(\StratProfile, \ConditionValue) = \Utility_\PlayerIndex(\StratProfile) + \ConditionValue}$, where $\smash{\ConditionValue \sim U(\nicefrac{-\NoiseCondition}{2}, \nicefrac{\NoiseCondition}{2})}$ and i.i.d..
The parameter $\NoiseCondition$ controls the magnitude of the noise. 
This simple model is meant to capture situations 
where it is hard or even impossible to make accurate distributional assumptions about noise.

\subsection{Experimental Results}
\paragraph{$\epsilon$ as a function of the number of samples.}
For these experiments, we ran the \GS{} algorithm using 1-ERA bounds on 200 games drawn at random from $\smash{\RandomCongestionGame(5, 5, 2)}$ using different levels of noise, namely $\NoiseCondition \in \{2,5,10\}$.
In~\Cref{fig:1ERA-GS-Performance} (a), we report 95\% confidence intervals around the error tolerance $\epsilon$, as a function of the number of samples $m$.
We see that the average error $\epsilon$ decreases roughly as $\nicefrac{1}{\sqrt{m}}$.
These findings are typical of games drawn from distribution $\smash{\RandomCongestionGame}$ with $\smash{\NumberOfPlayers \in \{2,\ldots,10\}}$, $m = \{2,\ldots,5\}$, and $k \in \{1,\ldots,4\}$.

Similar findings were obtained for games drawn from distribution $\RandomGame$.

\begin{figure}[ht]
    \centering
    \begin{subfigure}[b]{0.475\textwidth}
        \centering
        \vspace{0.25cm}
        \includegraphics[width=\textwidth]{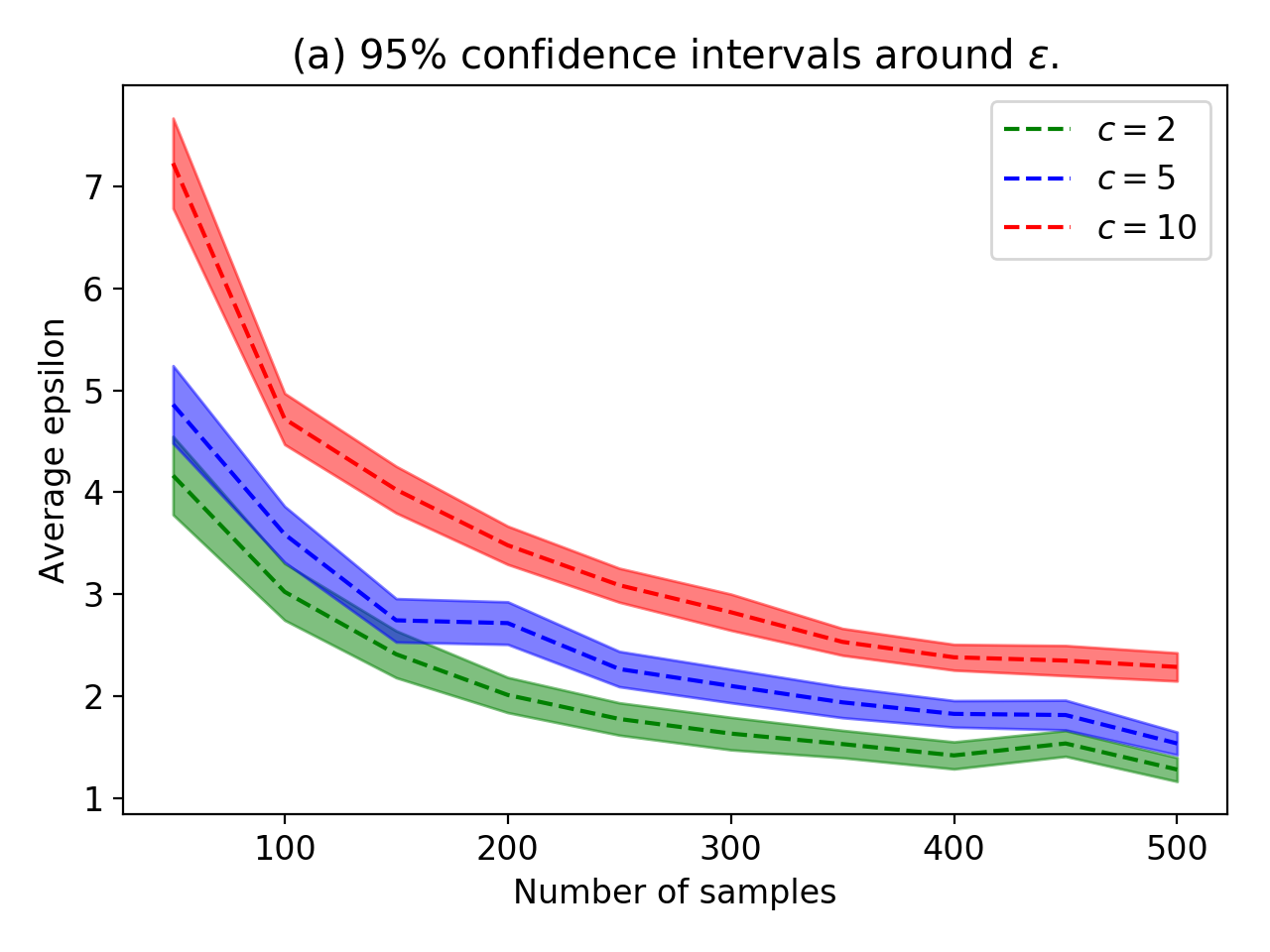}
    \end{subfigure}    
    \mbox{}\hspace{5mm}
    \begin{subfigure}[b]{0.475\textwidth}
        \centering
        \vspace{0.25cm}
        \includegraphics[width=\textwidth]{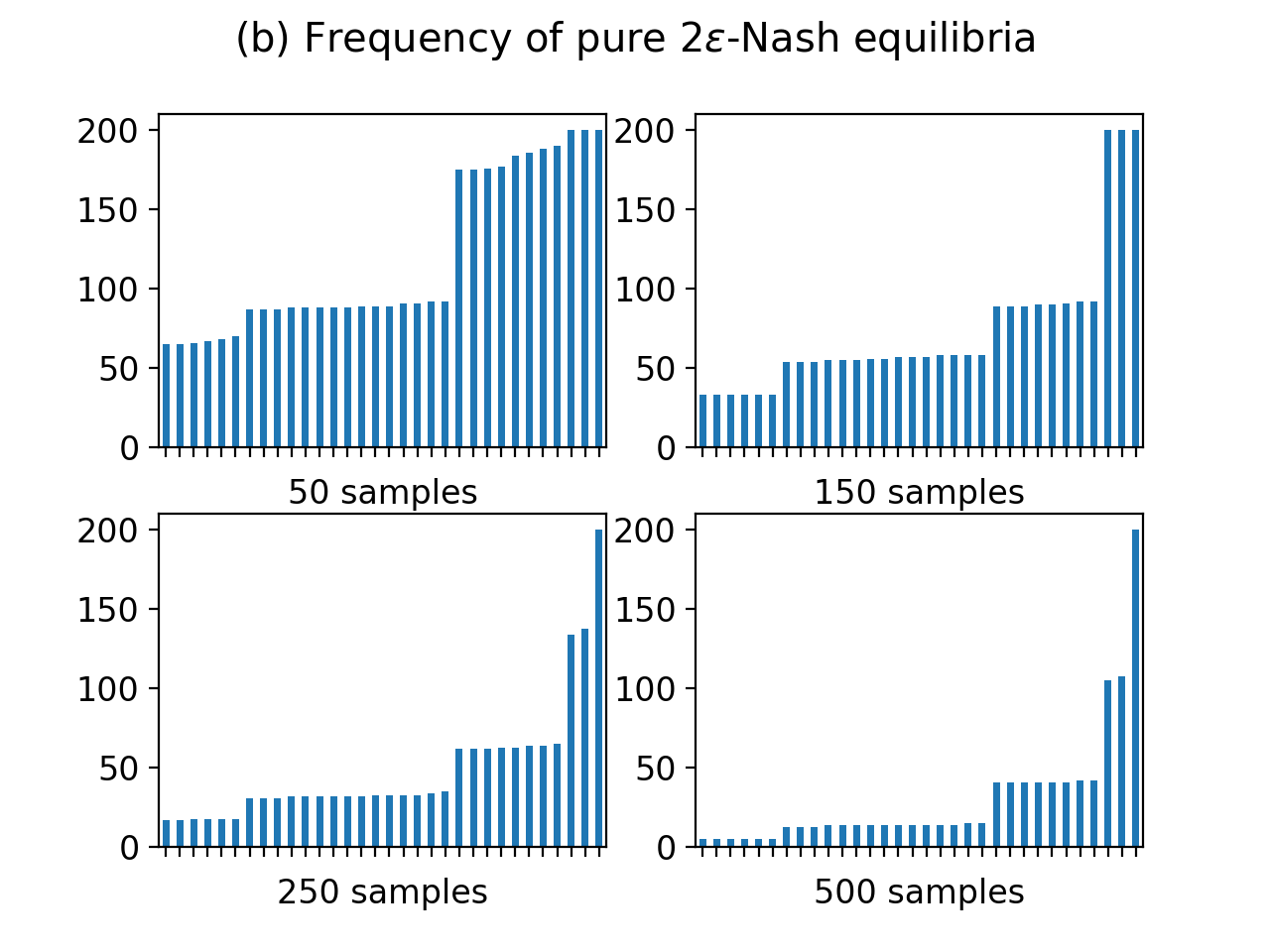}
    \end{subfigure}
    \vspace{-0.75cm}
    \caption{Performance of \GS{} with 1-ERA bounds.
    (a) 95\% confidence intervals around $\epsilon$ for games drawn from $\smash{\RandomCongestionGame(5, 5, 2)}$ using different levels of noise, namely $\NoiseCondition \in \{2,5,10\}$.
    (b) Frequency of pure $2\epsilon$-Nash equilibria in one game with 32 strategy profiles drawn from distribution $\RandomCongestionGame(5,5,2)$ with noise parameter $\NoiseCondition = 2$.}
    \label{fig:1ERA-GS-Performance}
\end{figure}

\paragraph{$2\epsilon$-Nash learning as a function of the number of samples.}
In these experiments, we ran the GS algorithm using 1-ERA bounds on one game with 32 strategy profiles drawn from distribution $\RandomCongestionGame(5,5,2)$ with noise parameter $\NoiseCondition = 2$. This game has a unique pure Nash equilibrium.

\Cref{fig:1ERA-GS-Performance} (b) depicts four histograms, corresponding to four different sample sizes, each one plotting the frequencies of strategy profiles deemed pure $2\epsilon$-Nash equilibria, using exhaustive search and applying Theorem~\ref{thm:gs-guarantee} to the error returned by \GS. The profiles not shown had zero frequency. The unique pure Nash equilibrium is correctly identified in all cases. Moreover, as the number of samples increases, the frequency of false positives (i.e., profiles that are not Nash but are deemed so by \GS) decreases.

\paragraph{Empirical success rate.}
In these experiments, we ran the \GS{} algorithm using 1-ERA bounds and Hoeffding's bound with a Bonferroni correction on 200 games drawn at random from $\RandomCongestionGame(3, 3, 2)$ and another 200 games drawn at random from $\RandomGame(3,3)$, both with $\NoiseCondition = 5$. 

In~\Cref{fig:GSvsPSP} (a), we report 95\% confidence intervals around the empirical success rate, meaning the ratio of the number of times the algorithm's output satisfies $\Nash(\InducedGame{\ConditionDistribution})  
    \subseteq 
\Nash_{2\epsilon}(\hat{\GameTuple}_{\Samples})$ and 
$\Nash_{2\epsilon}(\hat{\GameTuple}_{\Samples}) 
    \subseteq 
\Nash_{4\epsilon}(\InducedGame{\ConditionDistribution})$ to the total number of times it executes. As per our guarantees, this rate should be above the black line $1 - \delta$ (and it is). Note that we test these containments only for pure $\epsilon$-Nash equilibria.

To test the extent to which the bounds computed by \GS{} are tight, we contract the value of $\epsilon$ output by \GS{} by a factor $\rho \in \{1.0, 0.875, 0.75, 0.625, 0.5\}$. Even so, \GS{} consistently meets its $1-\delta$ guarantee, and in these experiments design, it is very often the case that guarantees 
are met with empirical probability close to 1, except when $\rho = 0.5$ and $\delta < 0.15$. 

\paragraph{Evaluation of \PSP} 
We now compare the performance of \GS{} to \PSP{} with a doubling schedule initialized at 100 samples. (As usual, we fix $\delta = 0.1$.) Since our experimental testbed only includes small games, we present results only for the (Bonferroni) correction-based variants of these algorithms, leaving experiments using Rademacher averages for future work.
Our procedure to fairly compare them is as follows. We first run \PSP{} with $\epsilon = 0$, thereby forcing the algorithm to run to completion. In this way, we obtain the error rate corresponding to \PSP's final \GS{} invocation. We call this quantity $\smash{\hat{\epsilon}}_{\PSP}$, and we denote by $M_{\PSP}$ the total number of samples used upon termination.
Now, since \GS{} takes in a number of samples which is used uniformly across all players and strategy profiles %
to produce error rate $\smash{\hat{\epsilon}}_{\GS}$, we compute $\smash{\hat{\epsilon}}_{\GS}$ 
assuming \GS{} runs on $\nicefrac{M_{\PSP}}{\SizeOfGame{\GameTuple}}$ samples. In other words, we assume a uniform budget of $\nicefrac{M_{\PSP}}{\SizeOfGame{\GameTuple}}$ samples per $\Utility_\PlayerIndex(\StratProfile)$. This strategy guarantees both algorithms use the same total number of samples. 

\Cref{fig:GSvsPSP} (b) depicts $\smash{\hat{\epsilon}}$, the approximation guarantee obtained for both algorithms, as a function of the number of samples, summarizing the results of experiments with 200 games of different sizes drawn from a $\RandomGame$ distribution, with $|P| \in \{ 2, 3, 4, 5 \}$ and $k \in \{ 2, 3, 4, 5 \}$.%
\footnote{Similar behavior was observed for $\RandomCongestionGame$ with with $|P| \in \{ 5, 6, 7, 8 \}$, $|E| \in \{ 2, 3, 4, 5 \}$, and $k = 2$.}
Game size, $\SizeOfGame{\GameTuple}$, is the number of parameters, meaning the number of utilities to be estimated, which equals the number of players times the number of strategy profiles.
The plot shows $\smash{\hat{\epsilon}}_{\PSP}$ in orange and $\smash{\hat{\epsilon}}_{\GS}$ in grey, on a logarithmic scale, thereby demonstrating that, for the same number of samples, \PSP{} yields better approximations: i.e., smaller values of $\smash{\hat{\epsilon}}$. Note, however, that \PSP{} relies on having sufficiently many strategy profiles to prune in the first place. For smaller games $\SizeOfGame{\GameTuple} \in \{8, 24\}$ and relatively small values of $\epsilon$, \GS{} produces slightly better approximations, because there is relatively little opportunity for pruning. For larger games $\SizeOfGame{\GameTuple} \in \{324, 1024\}$, \PSP{} consistently yields better approximations, and moreover, the error rate improves as a function of the size of the game.

\begin{figure}
    \centering
    \begin{subfigure}[b]{0.475\textwidth}
        \centering
        \vspace{-0.3cm}
        \includegraphics[width=\textwidth]{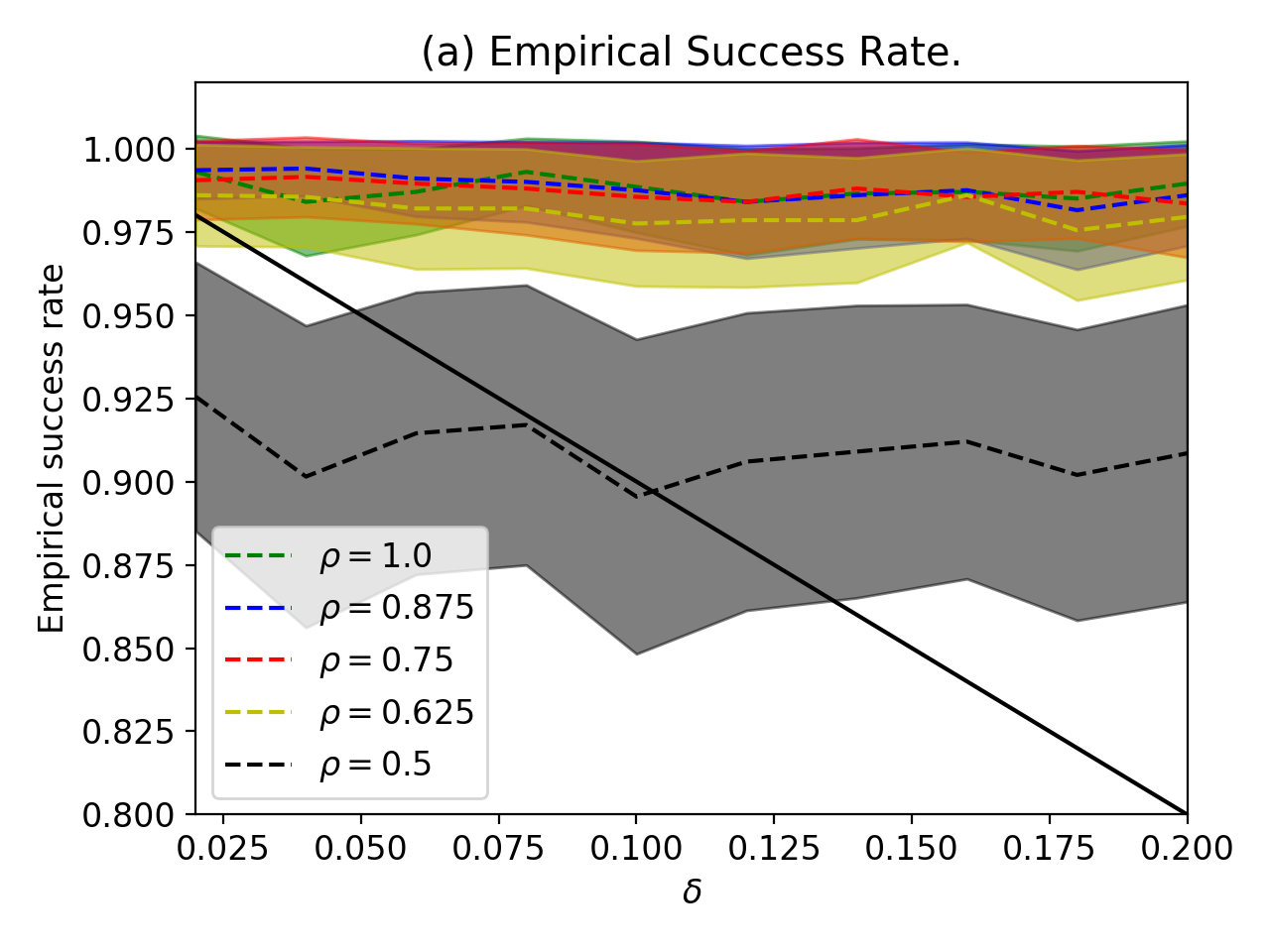}
        \vspace{-0.25cm}
    \end{subfigure}
    \mbox{}\hspace{5mm}
    \begin{subfigure}[b]{0.475\textwidth}
        \centering
        \vspace{-0.3cm}
        \includegraphics[width=\textwidth]{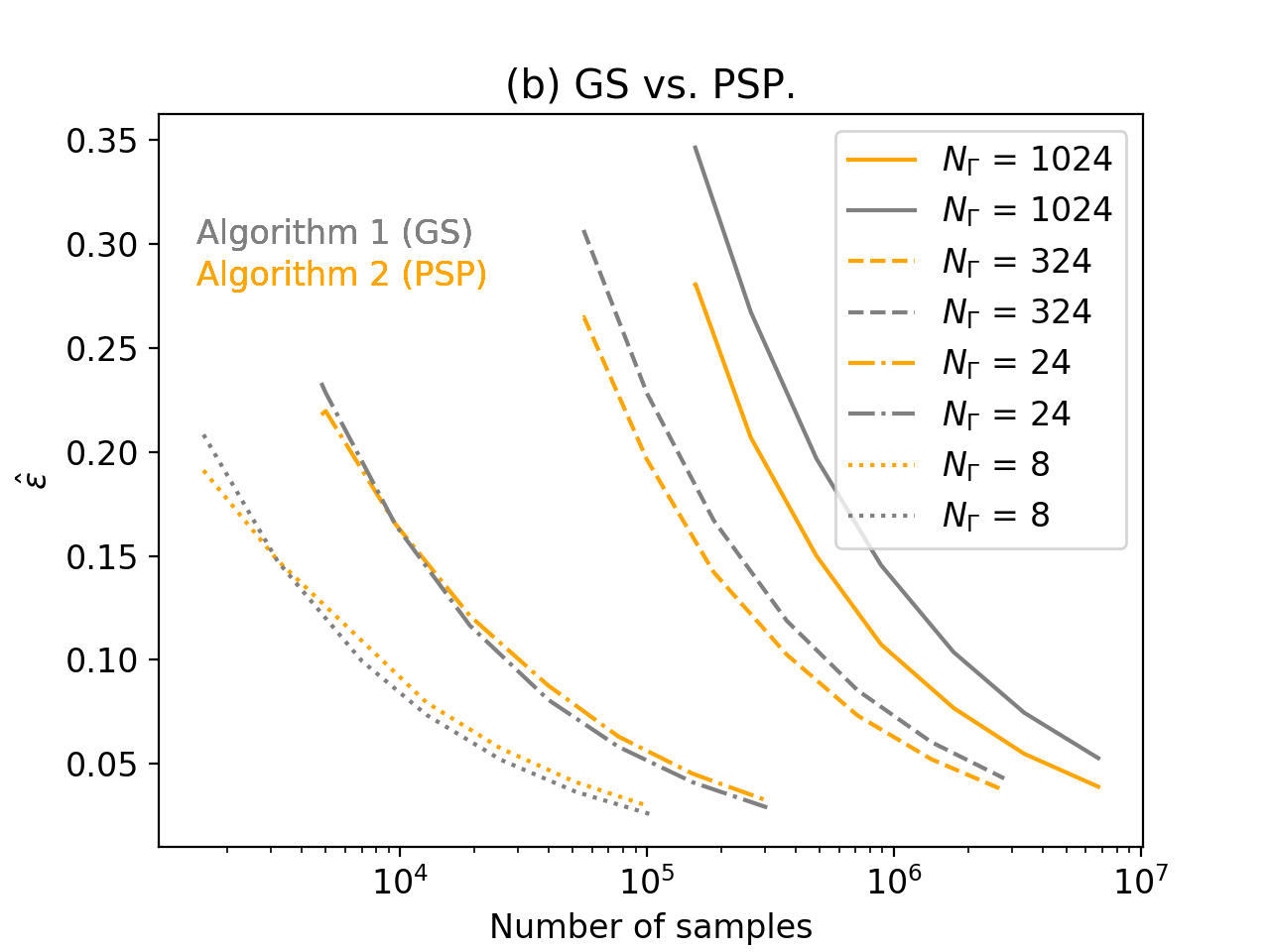}
        \vspace{-0.25cm}
    \end{subfigure}
    \vspace{-0.75cm}
    \caption{GS vs.\ PSP. 
    (a) Empirical success rate for 200 games drawn at random from $\RandomCongestionGame(3, 3, 2)$ and another 200 games drawn at random from $\RandomGame(3,3)$, both with $\NoiseCondition = 5$. 
    (b) GS vs.\ PSP for 200 games of different sizes drawn from a $\RandomGame$ distribution, with $|P| \in \{ 2, 3, 4, 5 \}$ and $k \in \{ 2, 3, 4, 5 \}$.}
    \label{fig:GSvsPSP}
\end{figure}

\paragraph{Future experimental work}
Thus far, we have only experimented with learning pure strategy $\epsilon$-Nash equilibria.
Our experiments are small enough that we can use exhaustive search to determine ground truth. Computing mixed strategy Nash equilibria is intractable (PPAD-complete)~\cite{daskalakis2009complexity};
accordingly, estimating them remains a challenge.

Having said that, a potentially interesting statistical vs.\ computational trade-off exists that warrants further investigation 
when the goal is to estimate mixed strategy equilibria.  While fully computing $\Rationalizable_{\epsilon}(\GameTuple)$ may incur \emph{statistical\/} savings by avoiding sampling profiles that are deemed %
non-equilibria early on, iteratively computing $\Rationalizable_{\epsilon}(\GameTuple)$ comes at a \emph{computational\/} cost.

The experiments we have run thus far, in simple, synthetic environments, were designed to illustrate the advantage of our algorithms in a controlled setting. In future work, we plan to study more complicated games, including large enough games so that bounds that are game-size dependent---such as those that require correction procedures, cf.\ \Cref{thm:Hoeffding}---are statistically intractable. Since our algorithms are agnostic to the noise distribution, we will also test their robustness to other forms of noise, such as multiplicative and unbounded noise.

\section{Summary and Future Directions}

This work is a contribution to the 
theoretical literature on empirical game-theoretic analysis, a methodology for the analysis of multi-agent systems.
One important future-work application for our techniques is equilibria estimation in meta-games.
Meta-games are simplified versions of intractably large games, where, instead of modeling every possible strategy an agent might implement, one analyzes a game with a substantially reduced set of strategies, each of which is usually given by a complicated algorithmic procedure (i.e., a heuristic). For example, 
one might analyze a reduced version of the game of Starcraft~\cite{tuyls2018generalised} where agents play according to higher-level strategies given by reinforcement learning algorithms: e.g., variants of AlphaGo~\cite{silver2016mastering}.
Simulation is in order; and since each run of the game can result in either agent winning (depending on various stochastic elements, including possibly the agents' strategies), one can only obtain noisy utilities. Our techniques are directly applicable to the construction of empirical meta-games, and provide guarantees on the quality of the equilibria of the corresponding simulation-based meta-games.

Our methodology is also applicable in the area of \mydef{empirical mechanism design}~\cite{vorobeychik2006empirical}, where the mechanism (game) designer wishes to optimize the rules of a game so that the ensuing equilibria achieve certain goals. For example, a network designer might want to minimize \emph{congestion} assuming selfish agents~\cite{roughgarden2002bad}.
In related work, we extend our methodology so that we can learn empirically optimal mechanisms, under appropriate assumptions~\cite{areyan2019pmduu}.

While uniform $\epsilon$-approximations allow us to estimate equilibria, not all properties of games can be estimated using approximations.
For example, while a uniform $\epsilon$-approximation allows us to estimate welfare to within $\NumberOfPlayers\epsilon$ additive error, they are not sufficient to estimate the welfare of the \emph{worst\/} (pure or mixed) Nash equilibrium, because the worst Nash equilibrium in $\GameTuple'$ could be arbitrarily better or worse than the worst equilibrium in $\GameTuple$, and even the worst $2\epsilon$-Nash equilibrium in $\GameTuple'$ could be arbitrarily worse than the worst Nash equilibrium in $\GameTuple$.  Thus quantities like the \mydef{price of anarchy} are also difficult to estimate, even given a uniform $\epsilon$-approximation.

Traditional approaches to empirical game-theoretic analysis, including past theory that employed statistical analysis, applied solely to finite games.
Likewise, our learning algorithms depend on an enumeration of all agents' utilities at all pure strategy profiles.
Our Rademacher bounds, however, do not require the finite assumption; they apply even to 
games with an infinite number of utilities, including countably infinite (or uncountably infinite, subject to standard measurability concerns) agents and/or pure strategies.
In future work, we will extend our methods to use function approximators in infinite games, and establish uniform convergence from only a finite sample, under alternative assumptions: e.g., Lipschitz utilities in games with finitely many agents and bounded, continuous pure strategy spaces.
Finally, we also hope to replace the generic McDiarmid concentration inequalities used to show the Rademacher bounds with tighter concentration inequalities, after making additional assumptions, such as uniformly bounded variance.

\begin{acks}
This work was supported in part by NSF Grant CMMI-1761546, the NEC-AIST AI Cooperative Research Laboratory, and NSF Award RI-1813444.
\end{acks}

\bibliographystyle{ACM-Reference-Format}
\bibliography{bibliography}


\begin{thebibliography}{00}


\ifx \showCODEN    \undefined \def \showCODEN     #1{\unskip}     \fi
\ifx \showDOI      \undefined \def \showDOI       #1{#1}\fi
\ifx \showISBNx    \undefined \def \showISBNx     #1{\unskip}     \fi
\ifx \showISBNxiii \undefined \def \showISBNxiii  #1{\unskip}     \fi
\ifx \showISSN     \undefined \def \showISSN      #1{\unskip}     \fi
\ifx \showLCCN     \undefined \def \showLCCN      #1{\unskip}     \fi
\ifx \shownote     \undefined \def \shownote      #1{#1}          \fi
\ifx \showarticletitle \undefined \def \showarticletitle #1{#1}   \fi
\ifx \showURL      \undefined \def \showURL       {\relax}        \fi
\providecommand\bibfield[2]{#2}
\providecommand\bibinfo[2]{#2}
\providecommand\natexlab[1]{#1}
\providecommand\showeprint[2][]{arXiv:#2}

\bibitem[\protect\citeauthoryear{Babichenko}{Babichenko}{[n. d.]}]%
        {babichenko2016query}
\bibfield{author}{\bibinfo{person}{Yakov Babichenko}.} \bibinfo{year}{[n.
  d.]}\natexlab{}.
\newblock \showarticletitle{Query complexity of approximate Nash equilibria}.
\newblock \bibinfo{journal}{{\em Journal of the ACM (JACM)\/}}
  (\bibinfo{year}{[n. d.]}).
\newblock


\bibitem[\protect\citeauthoryear{Bartlett and Mendelson}{Bartlett and
  Mendelson}{2002}]%
        {bartlett2002rademacher}
\bibfield{author}{\bibinfo{person}{Peter~L Bartlett} {and}
  \bibinfo{person}{Shahar Mendelson}.} \bibinfo{year}{2002}\natexlab{}.
\newblock \showarticletitle{Rademacher and Gaussian complexities: Risk bounds
  and structural results}.
\newblock \bibinfo{journal}{{\em Journal of Machine Learning Research\/}}
  \bibinfo{volume}{3}, \bibinfo{number}{Nov} (\bibinfo{year}{2002}),
  \bibinfo{pages}{463--482}.
\newblock


\bibitem[\protect\citeauthoryear{Boucheron, Bousquet, and Lugosi}{Boucheron
  et~al\mbox{.}}{2005}]%
        {boucheron2005theory}
\bibfield{author}{\bibinfo{person}{St{\'e}phane Boucheron},
  \bibinfo{person}{Olivier Bousquet}, {and} \bibinfo{person}{G{\'a}bor
  Lugosi}.} \bibinfo{year}{2005}\natexlab{}.
\newblock \showarticletitle{Theory of classification: A survey of some recent
  advances}.
\newblock \bibinfo{journal}{{\em ESAIM: probability and statistics\/}}
  \bibinfo{volume}{9} (\bibinfo{year}{2005}), \bibinfo{pages}{323--375}.
\newblock


\bibitem[\protect\citeauthoryear{Burch, Schmid, Morav{\v{c}}{\'\i}k, and
  Bowling}{Burch et~al\mbox{.}}{2016}]%
        {burch2016aivat}
\bibfield{author}{\bibinfo{person}{Neil Burch}, \bibinfo{person}{Martin
  Schmid}, \bibinfo{person}{Matej Morav{\v{c}}{\'\i}k}, {and}
  \bibinfo{person}{Michael Bowling}.} \bibinfo{year}{2016}\natexlab{}.
\newblock \showarticletitle{Aivat: A new variance reduction technique for agent
  evaluation in imperfect information games}.
\newblock \bibinfo{journal}{{\em arXiv preprint arXiv:1612.06915\/}}
  (\bibinfo{year}{2016}).
\newblock


\bibitem[\protect\citeauthoryear{Christodoulou and Koutsoupias}{Christodoulou
  and Koutsoupias}{2005}]%
        {christodoulou2005price}
\bibfield{author}{\bibinfo{person}{George Christodoulou} {and}
  \bibinfo{person}{Elias Koutsoupias}.} \bibinfo{year}{2005}\natexlab{}.
\newblock \showarticletitle{The price of anarchy of finite congestion games}.
  In \bibinfo{booktitle}{{\em Proceedings of the thirty-seventh annual ACM
  symposium on Theory of computing}}. ACM, \bibinfo{pages}{67--73}.
\newblock


\bibitem[\protect\citeauthoryear{Daskalakis, Goldberg, and
  Papadimitriou}{Daskalakis et~al\mbox{.}}{2009}]%
        {daskalakis2009complexity}
\bibfield{author}{\bibinfo{person}{Constantinos Daskalakis},
  \bibinfo{person}{Paul~W Goldberg}, {and} \bibinfo{person}{Christos~H
  Papadimitriou}.} \bibinfo{year}{2009}\natexlab{}.
\newblock \showarticletitle{The complexity of computing a Nash equilibrium}.
\newblock \bibinfo{journal}{{\it SIAM J. Comput.}} \bibinfo{volume}{39},
  \bibinfo{number}{1} (\bibinfo{year}{2009}), \bibinfo{pages}{195--259}.
\newblock


\bibitem[\protect\citeauthoryear{Efron and Tibshirani}{Efron and
  Tibshirani}{1993}]%
        {EfroTibs93}
\bibfield{author}{\bibinfo{person}{Bradley Efron} {and}
  \bibinfo{person}{Robert~J. Tibshirani}.} \bibinfo{year}{1993}\natexlab{}.
\newblock \bibinfo{booktitle}{{\em An Introduction to the Bootstrap}}.
\newblock Number~57 in \bibinfo{series}{Monographs on Statistics and Applied
  Probability}. \bibinfo{publisher}{Chapman \& Hall/CRC},
  \bibinfo{address}{Boca Raton, Florida, USA}.
\newblock


\bibitem[\protect\citeauthoryear{Elomaa and K{\"a}{\"a}ri{\"a}inen}{Elomaa and
  K{\"a}{\"a}ri{\"a}inen}{2002}]%
        {elomaa2002progressive}
\bibfield{author}{\bibinfo{person}{Tapio Elomaa} {and} \bibinfo{person}{Matti
  K{\"a}{\"a}ri{\"a}inen}.} \bibinfo{year}{2002}\natexlab{}.
\newblock \showarticletitle{Progressive {R}ademacher sampling}. In
  \bibinfo{booktitle}{{\em AAAI/IAAI}}. \bibinfo{pages}{140--145}.
\newblock


\bibitem[\protect\citeauthoryear{Fearnley, Gairing, Goldberg, and
  Savani}{Fearnley et~al\mbox{.}}{2015}]%
        {fearnley2015learning}
\bibfield{author}{\bibinfo{person}{John Fearnley}, \bibinfo{person}{Martin
  Gairing}, \bibinfo{person}{Paul~W Goldberg}, {and} \bibinfo{person}{Rahul
  Savani}.} \bibinfo{year}{2015}\natexlab{}.
\newblock \showarticletitle{Learning equilibria of games via payoff queries}.
\newblock \bibinfo{journal}{{\em The Journal of Machine Learning Research\/}}
  \bibinfo{volume}{16}, \bibinfo{number}{1} (\bibinfo{year}{2015}),
  \bibinfo{pages}{1305--1344}.
\newblock


\bibitem[\protect\citeauthoryear{Gibbons}{Gibbons}{1992}]%
        {gibbons1992game}
\bibfield{author}{\bibinfo{person}{Robert Gibbons}.}
  \bibinfo{year}{1992}\natexlab{}.
\newblock \bibinfo{booktitle}{{\em Game theory for applied economists}}.
\newblock \bibinfo{publisher}{Princeton University Press}.
\newblock


\bibitem[\protect\citeauthoryear{Jecmen, Brinkman, and Sinha}{Jecmen
  et~al\mbox{.}}{2018}]%
        {jecmen2018bounding}
\bibfield{author}{\bibinfo{person}{S. Jecmen}, \bibinfo{person}{E. Brinkman},
  {and} \bibinfo{person}{A. Sinha}.} \bibinfo{year}{2018}\natexlab{}.
\newblock \showarticletitle{Bounding Regret in Simulated Games}. In
  \bibinfo{booktitle}{{\em Exploration in RL at ICML2018}}.
\newblock


\bibitem[\protect\citeauthoryear{Jordan, Kiekintveld, and Wellman}{Jordan
  et~al\mbox{.}}{2007}]%
        {jordan2007empirical}
\bibfield{author}{\bibinfo{person}{Patrick~R Jordan},
  \bibinfo{person}{Christopher Kiekintveld}, {and} \bibinfo{person}{Michael~P
  Wellman}.} \bibinfo{year}{2007}\natexlab{}.
\newblock \showarticletitle{Empirical game-theoretic analysis of the TAC supply
  chain game}. In \bibinfo{booktitle}{{\em Proceedings of the 6th international
  joint conference on Autonomous agents and multiagent systems}}. ACM,
  \bibinfo{pages}{193}.
\newblock


\bibitem[\protect\citeauthoryear{Jordan, Vorobeychik, and Wellman}{Jordan
  et~al\mbox{.}}{2008}]%
        {jordan2008searching}
\bibfield{author}{\bibinfo{person}{Patrick~R Jordan}, \bibinfo{person}{Yevgeniy
  Vorobeychik}, {and} \bibinfo{person}{Michael~P Wellman}.}
  \bibinfo{year}{2008}\natexlab{}.
\newblock \showarticletitle{Searching for approximate equilibria in empirical
  games}. In \bibinfo{booktitle}{{\em Proceedings of the 7th international
  joint conference on Autonomous agents and multiagent systems-Volume 2}}.
  International Foundation for Autonomous Agents and Multiagent Systems,
  \bibinfo{pages}{1063--1070}.
\newblock


\bibitem[\protect\citeauthoryear{Jordan and Wellman}{Jordan and
  Wellman}{2010}]%
        {jordan2010designing}
\bibfield{author}{\bibinfo{person}{Patrick~R Jordan} {and}
  \bibinfo{person}{Michael~P Wellman}.} \bibinfo{year}{2010}\natexlab{}.
\newblock \showarticletitle{Designing an ad auctions game for the trading agent
  competition}.
\newblock In \bibinfo{booktitle}{{\em Agent-Mediated Electronic Commerce.
  Designing Trading Strategies and Mechanisms for Electronic Markets}}.
  \bibinfo{publisher}{Springer}.
\newblock


\bibitem[\protect\citeauthoryear{K{\"a}{\"a}ri{\"a}inen, Malinen, and
  Elomaa}{K{\"a}{\"a}ri{\"a}inen et~al\mbox{.}}{2004}]%
        {kaariainen2004selective}
\bibfield{author}{\bibinfo{person}{Matti K{\"a}{\"a}ri{\"a}inen},
  \bibinfo{person}{Tuomo Malinen}, {and} \bibinfo{person}{Tapio Elomaa}.}
  \bibinfo{year}{2004}\natexlab{}.
\newblock \showarticletitle{Selective {R}ademacher penalization and reduced
  error pruning of decision trees}.
\newblock \bibinfo{journal}{{\em Journal of Machine Learning Research\/}}
  \bibinfo{volume}{5}, \bibinfo{number}{Sep} (\bibinfo{year}{2004}),
  \bibinfo{pages}{1107--1126}.
\newblock


\bibitem[\protect\citeauthoryear{Ketter, Peters, and Collins}{Ketter
  et~al\mbox{.}}{2013}]%
        {ketter2013autonomous}
\bibfield{author}{\bibinfo{person}{Wolfgang Ketter}, \bibinfo{person}{Markus
  Peters}, {and} \bibinfo{person}{John Collins}.}
  \bibinfo{year}{2013}\natexlab{}.
\newblock \showarticletitle{Autonomous Agents in Future Energy Markets: The
  2012 Power Trading Agent Competition.}. In \bibinfo{booktitle}{{\em AAAI}}.
\newblock


\bibitem[\protect\citeauthoryear{Koltchinskii}{Koltchinskii}{2001}]%
        {koltchinskii2001rademacher}
\bibfield{author}{\bibinfo{person}{Vladimir Koltchinskii}.}
  \bibinfo{year}{2001}\natexlab{}.
\newblock \showarticletitle{Rademacher penalties and structural risk
  minimization}.
\newblock \bibinfo{journal}{{\em IEEE Transactions on Information Theory\/}}
  \bibinfo{volume}{47}, \bibinfo{number}{5} (\bibinfo{year}{2001}),
  \bibinfo{pages}{1902--1914}.
\newblock


\bibitem[\protect\citeauthoryear{Massart}{Massart}{2000}]%
        {massart2000some}
\bibfield{author}{\bibinfo{person}{Pascal Massart}.}
  \bibinfo{year}{2000}\natexlab{}.
\newblock \showarticletitle{Some applications of concentration inequalities to
  statistics}. In \bibinfo{booktitle}{{\em Annales-Faculte des Sciences
  Toulouse Mathematiques}}, Vol.~\bibinfo{volume}{9}. Universit{\'e} Paul
  Sabatier, \bibinfo{pages}{245--303}.
\newblock


\bibitem[\protect\citeauthoryear{McDiarmid}{McDiarmid}{1989}]%
        {mcdiarmid1989method}
\bibfield{author}{\bibinfo{person}{Colin McDiarmid}.}
  \bibinfo{year}{1989}\natexlab{}.
\newblock \showarticletitle{On the method of bounded differences}.
\newblock \bibinfo{journal}{{\em Surveys in combinatorics\/}}
  \bibinfo{volume}{141}, \bibinfo{number}{1} (\bibinfo{year}{1989}),
  \bibinfo{pages}{148--188}.
\newblock


\bibitem[\protect\citeauthoryear{Mitzenmacher}{Mitzenmacher}{2004}]%
        {mitzenmacher2004brief}
\bibfield{author}{\bibinfo{person}{Michael Mitzenmacher}.}
  \bibinfo{year}{2004}\natexlab{}.
\newblock \showarticletitle{A brief history of generative models for power law
  and lognormal distributions}.
\newblock \bibinfo{journal}{{\em Internet mathematics\/}} \bibinfo{volume}{1},
  \bibinfo{number}{2} (\bibinfo{year}{2004}), \bibinfo{pages}{226--251}.
\newblock


\bibitem[\protect\citeauthoryear{Mitzenmacher and Upfal}{Mitzenmacher and
  Upfal}{2017}]%
        {mitzenmacher2017probability}
\bibfield{author}{\bibinfo{person}{Michael Mitzenmacher} {and}
  \bibinfo{person}{Eli Upfal}.} \bibinfo{year}{2017}\natexlab{}.
\newblock \bibinfo{booktitle}{{\em Probability and computing: Randomization and
  probabilistic techniques in algorithms and data analysis\/}
  (\bibinfo{edition}{2} ed.)}.
\newblock \bibinfo{publisher}{Cambridge {U}niversity {P}ress}.
\newblock


\bibitem[\protect\citeauthoryear{Nash}{Nash}{1950}]%
        {nash1950equilibrium}
\bibfield{author}{\bibinfo{person}{John~F Nash}.}
  \bibinfo{year}{1950}\natexlab{}.
\newblock \showarticletitle{Equilibrium points in n-person games}.
\newblock \bibinfo{journal}{{\em Proceedings of the national academy of
  sciences\/}} (\bibinfo{year}{1950}).
\newblock


\bibitem[\protect\citeauthoryear{Picheny, Binois, and Habbal}{Picheny
  et~al\mbox{.}}{2016}]%
        {picheny2016bayesian}
\bibfield{author}{\bibinfo{person}{Victor Picheny}, \bibinfo{person}{Mickael
  Binois}, {and} \bibinfo{person}{Abderrahmane Habbal}.}
  \bibinfo{year}{2016}\natexlab{}.
\newblock \showarticletitle{A Bayesian optimization approach to find Nash
  equilibria}.
\newblock \bibinfo{journal}{{\em arXiv preprint arXiv:1611.02440\/}}
  (\bibinfo{year}{2016}).
\newblock


\bibitem[\protect\citeauthoryear{Riondato and Upfal}{Riondato and
  Upfal}{2015}]%
        {riondato2015mining}
\bibfield{author}{\bibinfo{person}{Matteo Riondato} {and} \bibinfo{person}{Eli
  Upfal}.} \bibinfo{year}{2015}\natexlab{}.
\newblock \showarticletitle{Mining frequent itemsets through progressive
  sampling with {R}ademacher averages}. In \bibinfo{booktitle}{{\em Proceedings
  of the 21th ACM SIGKDD International Conference on Knowledge Discovery and
  Data Mining}}. ACM, \bibinfo{pages}{1005--1014}.
\newblock


\bibitem[\protect\citeauthoryear{Riondato and Upfal}{Riondato and
  Upfal}{2018}]%
        {riondato2018abra}
\bibfield{author}{\bibinfo{person}{Matteo Riondato} {and} \bibinfo{person}{Eli
  Upfal}.} \bibinfo{year}{2018}\natexlab{}.
\newblock \showarticletitle{ABRA: Approximating betweenness centrality in
  static and dynamic graphs with {R}ademacher averages}.
\newblock \bibinfo{journal}{{\em ACM Transactions on Knowledge Discovery from
  Data (TKDD)\/}} \bibinfo{volume}{12}, \bibinfo{number}{5}
  (\bibinfo{year}{2018}), \bibinfo{pages}{61}.
\newblock


\bibitem[\protect\citeauthoryear{Rosenthal}{Rosenthal}{1973}]%
        {rosenthal1973class}
\bibfield{author}{\bibinfo{person}{Robert~W Rosenthal}.}
  \bibinfo{year}{1973}\natexlab{}.
\newblock \showarticletitle{A class of games possessing pure-strategy Nash
  equilibria}.
\newblock \bibinfo{journal}{{\em International Journal of Game Theory\/}}
  \bibinfo{volume}{2}, \bibinfo{number}{1} (\bibinfo{year}{1973}),
  \bibinfo{pages}{65--67}.
\newblock


\bibitem[\protect\citeauthoryear{Roughgarden and Tardos}{Roughgarden and
  Tardos}{2002}]%
        {roughgarden2002bad}
\bibfield{author}{\bibinfo{person}{Tim Roughgarden} {and}
  \bibinfo{person}{{\'E}va Tardos}.} \bibinfo{year}{2002}\natexlab{}.
\newblock \showarticletitle{How bad is selfish routing?}
\newblock \bibinfo{journal}{{\em Journal of the ACM (JACM)\/}}
  \bibinfo{volume}{49}, \bibinfo{number}{2} (\bibinfo{year}{2002}),
  \bibinfo{pages}{236--259}.
\newblock


\bibitem[\protect\citeauthoryear{Silver, Huang, Maddison, Guez, Sifre, Van
  Den~Driessche, Schrittwieser, Antonoglou, Panneershelvam, Lanctot,
  et~al\mbox{.}}{Silver et~al\mbox{.}}{2016}]%
        {silver2016mastering}
\bibfield{author}{\bibinfo{person}{David Silver}, \bibinfo{person}{Aja Huang},
  \bibinfo{person}{Chris~J Maddison}, \bibinfo{person}{Arthur Guez},
  \bibinfo{person}{Laurent Sifre}, \bibinfo{person}{George Van Den~Driessche},
  \bibinfo{person}{Julian Schrittwieser}, \bibinfo{person}{Ioannis Antonoglou},
  \bibinfo{person}{Veda Panneershelvam}, \bibinfo{person}{Marc Lanctot},
  {et~al\mbox{.}}} \bibinfo{year}{2016}\natexlab{}.
\newblock \showarticletitle{Mastering the game of Go with deep neural networks
  and tree search}.
\newblock \bibinfo{journal}{{\em nature\/}} \bibinfo{volume}{529},
  \bibinfo{number}{7587} (\bibinfo{year}{2016}), \bibinfo{pages}{484}.
\newblock


\bibitem[\protect\citeauthoryear{Simon}{Simon}{1955}]%
        {simon1955class}
\bibfield{author}{\bibinfo{person}{Herbert~A Simon}.}
  \bibinfo{year}{1955}\natexlab{}.
\newblock \showarticletitle{On a class of skew distribution functions}.
\newblock \bibinfo{journal}{{\em Biometrika\/}} \bibinfo{volume}{42},
  \bibinfo{number}{3/4} (\bibinfo{year}{1955}), \bibinfo{pages}{425--440}.
\newblock


\bibitem[\protect\citeauthoryear{Tavares, Azpurua, Santos, and
  Chaimowicz}{Tavares et~al\mbox{.}}{2016}]%
        {tavares2016rock}
\bibfield{author}{\bibinfo{person}{Anderson Tavares}, \bibinfo{person}{Hector
  Azpurua}, \bibinfo{person}{Amanda Santos}, {and} \bibinfo{person}{Luiz
  Chaimowicz}.} \bibinfo{year}{2016}\natexlab{}.
\newblock \showarticletitle{Rock, paper, starcraft: Strategy selection in
  real-time strategy games}. In \bibinfo{booktitle}{{\em The Twelfth AAAI
  Conference on Artificial Intelligence and Interactive Digital Entertainment
  (AIIDE-16)}}.
\newblock


\bibitem[\protect\citeauthoryear{Tuyls, Perolat, Lanctot, Leibo, and
  Graepel}{Tuyls et~al\mbox{.}}{2018}]%
        {tuyls2018generalised}
\bibfield{author}{\bibinfo{person}{Karl Tuyls}, \bibinfo{person}{Julien
  Perolat}, \bibinfo{person}{Marc Lanctot}, \bibinfo{person}{Joel~Z Leibo},
  {and} \bibinfo{person}{Thore Graepel}.} \bibinfo{year}{2018}\natexlab{}.
\newblock \showarticletitle{A Generalised Method for Empirical Game Theoretic
  Analysis}.
\newblock \bibinfo{journal}{{\em arXiv preprint arXiv:1803.06376\/}}
  (\bibinfo{year}{2018}).
\newblock


\bibitem[\protect\citeauthoryear{Valiant}{Valiant}{1984}]%
        {valiant1984theory}
\bibfield{author}{\bibinfo{person}{Leslie~G Valiant}.}
  \bibinfo{year}{1984}\natexlab{}.
\newblock \showarticletitle{A theory of the learnable}.
\newblock \bibinfo{journal}{{\it Commun. ACM}} \bibinfo{volume}{27},
  \bibinfo{number}{11} (\bibinfo{year}{1984}), \bibinfo{pages}{1134--1142}.
\newblock


\bibitem[\protect\citeauthoryear{Viqueira, Cousins, Mohammad, and
  Greenwald}{Viqueira et~al\mbox{.}}{[n. d.]}]%
        {areyan2019pmduu}
\bibfield{author}{\bibinfo{person}{Enrique~Areyan Viqueira},
  \bibinfo{person}{Cyrus Cousins}, \bibinfo{person}{Yasser Mohammad}, {and}
  \bibinfo{person}{Amy Greenwald}.} \bibinfo{year}{[n. d.]}\natexlab{}.
\newblock \showarticletitle{Parametric Mechanism Design under Uncertainty}. In
  \bibinfo{booktitle}{{\em To Appear on the Thirty-Fifth Conference on
  Uncertainty in Artificial Intelligence, {UAI} 2019, Tel Aviv, Israel, July
  22-25, 2019}}.
\newblock


\bibitem[\protect\citeauthoryear{Vorobeychik}{Vorobeychik}{2010}]%
        {vorobeychik2010probabilistic}
\bibfield{author}{\bibinfo{person}{Yevgeniy Vorobeychik}.}
  \bibinfo{year}{2010}\natexlab{}.
\newblock \showarticletitle{Probabilistic analysis of simulation-based games}.
\newblock \bibinfo{journal}{{\em ACM Transactions on Modeling and Computer
  Simulation (TOMACS)\/}} \bibinfo{volume}{20}, \bibinfo{number}{3}
  (\bibinfo{year}{2010}), \bibinfo{pages}{16}.
\newblock


\bibitem[\protect\citeauthoryear{Vorobeychik, Kiekintveld, and
  Wellman}{Vorobeychik et~al\mbox{.}}{2006}]%
        {vorobeychik2006empirical}
\bibfield{author}{\bibinfo{person}{Yevgeniy Vorobeychik},
  \bibinfo{person}{Christopher Kiekintveld}, {and} \bibinfo{person}{Michael~P
  Wellman}.} \bibinfo{year}{2006}\natexlab{}.
\newblock \showarticletitle{Empirical mechanism design: methods, with
  application to a supply-chain scenario}. In \bibinfo{booktitle}{{\em
  Proceedings of the 7th ACM conference on Electronic commerce}}. ACM.
\newblock


\bibitem[\protect\citeauthoryear{Vorobeychik and Wellman}{Vorobeychik and
  Wellman}{2008}]%
        {vorobeychik2008stochastic}
\bibfield{author}{\bibinfo{person}{Yevgeniy Vorobeychik} {and}
  \bibinfo{person}{Michael~P Wellman}.} \bibinfo{year}{2008}\natexlab{}.
\newblock \showarticletitle{Stochastic search methods for Nash equilibrium
  approximation in simulation-based games}. In \bibinfo{booktitle}{{\em
  Proceedings of the 7th international joint conference on Autonomous agents
  and multiagent systems-Volume 2}}. International Foundation for Autonomous
  Agents and Multiagent Systems, \bibinfo{pages}{1055--1062}.
\newblock


\bibitem[\protect\citeauthoryear{Vorobeychik, Wellman, and Singh}{Vorobeychik
  et~al\mbox{.}}{2007}]%
        {vorobeychik2007learning}
\bibfield{author}{\bibinfo{person}{Yevgeniy Vorobeychik},
  \bibinfo{person}{Michael~P Wellman}, {and} \bibinfo{person}{Satinder Singh}.}
  \bibinfo{year}{2007}\natexlab{}.
\newblock \showarticletitle{Learning payoff functions in infinite games}.
\newblock \bibinfo{journal}{{\em Machine Learning\/}} \bibinfo{volume}{67},
  \bibinfo{number}{1-2} (\bibinfo{year}{2007}), \bibinfo{pages}{145--168}.
\newblock


\bibitem[\protect\citeauthoryear{Walsh, Parkes, and Das}{Walsh
  et~al\mbox{.}}{2003}]%
        {walsh2003choosing}
\bibfield{author}{\bibinfo{person}{William~E Walsh}, \bibinfo{person}{David~C
  Parkes}, {and} \bibinfo{person}{Rajarshi Das}.}
  \bibinfo{year}{2003}\natexlab{}.
\newblock \showarticletitle{Choosing samples to compute heuristic-strategy Nash
  equilibrium}. In \bibinfo{booktitle}{{\em AMEC}}. Springer,
  \bibinfo{pages}{109--123}.
\newblock


\bibitem[\protect\citeauthoryear{Wellman}{Wellman}{2006}]%
        {wellman2006methods}
\bibfield{author}{\bibinfo{person}{Michael~P Wellman}.}
  \bibinfo{year}{2006}\natexlab{}.
\newblock \showarticletitle{Methods for empirical game-theoretic analysis}. In
  \bibinfo{booktitle}{{\em AAAI}}. \bibinfo{pages}{1552--1556}.
\newblock


\bibitem[\protect\citeauthoryear{Wellman, Kim, and Duong}{Wellman
  et~al\mbox{.}}{2013}]%
        {wellman2013analyzing}
\bibfield{author}{\bibinfo{person}{Michael~P Wellman},
  \bibinfo{person}{Tae~Hyung Kim}, {and} \bibinfo{person}{Quang Duong}.}
  \bibinfo{year}{2013}\natexlab{}.
\newblock \showarticletitle{Analyzing incentives for protocol compliance in
  complex domains: A case study of introduction-based routing}.
\newblock \bibinfo{journal}{{\em arXiv preprint arXiv:1306.0388\/}}
  (\bibinfo{year}{2013}).
\newblock


\bibitem[\protect\citeauthoryear{Wiedenbeck}{Wiedenbeck}{2014}]%
        {Wiedenbeck:2014:AGT:2615731.2616156}
\bibfield{author}{\bibinfo{person}{Bryce Wiedenbeck}.}
  \bibinfo{year}{2014}\natexlab{}.
\newblock \showarticletitle{Approximate Game Theoretic Analysis for Large
  Simulation-based Games}. In \bibinfo{booktitle}{{\em Proceedings of the 2014
  International Conference on Autonomous Agents and Multi-agent Systems}} {\em
  (\bibinfo{series}{AAMAS '14})}. \bibinfo{publisher}{International Foundation
  for Autonomous Agents and Multiagent Systems}, \bibinfo{address}{Richland,
  SC}, \bibinfo{pages}{1745--1746}.
\newblock
\showISBNx{978-1-4503-2738-1}
\showURL{%
\url{http://dl.acm.org/citation.cfm?id=2615731.2616156}}


\bibitem[\protect\citeauthoryear{Wiedenbeck, Yang, and Wellman}{Wiedenbeck
  et~al\mbox{.}}{2018}]%
        {wiedenbeck2018regression}
\bibfield{author}{\bibinfo{person}{Bryce Wiedenbeck}, \bibinfo{person}{Fengjun
  Yang}, {and} \bibinfo{person}{Michael~P Wellman}.}
  \bibinfo{year}{2018}\natexlab{}.
\newblock \showarticletitle{A Regression Approach for Modeling Games with Many
  Symmetric Players}. In \bibinfo{booktitle}{{\em 32nd AAAI Conference on
  Artificial Intelligence}}.
\newblock


\end{thebibliography}

\pagebreak[4]

\appendix
\section*{Appendix}
\section{Proofs of Concentration Inequalities}
\label{sec:proofs}

\begin{proof}[Proof of \Cref{thm:Hoeffding}]

Applying a union bound to Hoeffding's inequality yields:
\begin{align*}
{
\mathbb{P} \Bigg( \sup_{(\PlayerIndex, \StratProfile) \in \Indices}
\Bigg \lvert \Expwrt{\ConditionValue \distributed \ConditionDistribution} {\Utility_\PlayerIndex (\StratProfile; \SamplePoint)} - \frac{1}{\NumberOfSamples}
\sum_{\SampleIndex=1}^\NumberOfSamples \Utility_\PlayerIndex (\StratProfile; \SamplePoint_\SampleIndex) \Bigg \rvert \le \epsilon \Bigg)} \geq 
1 - \sum_{(\PlayerIndex, \StratProfile) \in \Indices} 2e^{\nicefrac{-2 \epsilon^2 \NumberOfSamples} {\UtilityRange^2}} = 1 - 2 \abs{\Indices} e^{\nicefrac{-2 \epsilon^2\NumberOfSamples} {\UtilityRange^2}}
\end{align*}

\noindent
Setting $\delta = 2\abs{\Indices}e^{\nicefrac{-2\epsilon^2\NumberOfSamples}{\UtilityRange^2}}$ and solving for $\epsilon$ yields the result.
\end{proof}

\begin{lemma}[Symmetrization Inequality]
\label{app:lemma:symmetrization}

Given expected game $\InducedGame{\ConditionDistribution}$, with utility function $\Utility(\cdot; \ConditionDistribution)$, sample $\Samples \distributed \ConditionDistribution^\NumberOfSamples$, and empirical utility function $\hat{\Utility}(\cdot; \Samples)$, it holds that
\[
\displaystyle \Expwrt{\Samples \distributed \ConditionDistribution^\NumberOfSamples}{\sup_{(\PlayerIndex, \StratProfile) \in \Indices}
\abs{\Utility_{\PlayerIndex} (\StratProfile; \ConditionDistribution) - \hat{\Utility}_{\PlayerIndex} (\StratProfile; \Samples)}} \leq 2 \RC {\NumberOfSamples}{\GameTuple, \Indices}{\ConditionDistribution} \enspace.
\]

\end{lemma}
\begin{proof}

In the proof below, suppose $\Samples' = (\SamplePoint'_1, \ldots, \SamplePoint'_\NumberOfSamples) \distributed \ConditionDistribution^{\NumberOfSamples}$.  The step labeled \textsc{See Above} follows, as $\SamplePoint'_j$ and $\SamplePoint_j$ are i.i.d., thus multiplying by a Rademacher random variable doesn't change the outer distribution over which the expectation is taken, as two equiprobable events are swapped.

\hspace{-0.25cm}\scalebox{0.9}{
\parbox{\linewidth}{
\begin{align*}
\hspace{0.75cm} & \hspace{-0.75cm} \Expwrt{\Samples}{\sup_{(\PlayerIndex, \StratProfile) \in \Indices} \abs{\Utility_{\PlayerIndex}(\StratProfile; \ConditionDistribution) - \hat{\Utility}_{\PlayerIndex} (\StratProfile; \Samples)}} \\
&= \Expwrt{\Samples}{\sup_{(\PlayerIndex, \StratProfile) \in \Indices} \abs{\Expwrt{ \! \Samples' \! }{\frac{1}{\NumberOfSamples} \! \sum_{\SampleIndex=1}^\NumberOfSamples\Utility_{\PlayerIndex}(\StratProfile ; \SamplePoint'_\SampleIndex)} - \frac{1}{\NumberOfSamples} \! \sum_{\SampleIndex=1}^\NumberOfSamples\Utility_{\PlayerIndex}(\StratProfile ; \SamplePoint_\SampleIndex)}} \hspace{-1cm} & \textsc{Definition of Expected Nfg} \\
&\leq \Expwrt{\Samples,\Samples'}{\sup_{(\PlayerIndex, \StratProfile) \in \Indices}  \abs{\frac{1}{\NumberOfSamples}\sum_{\SampleIndex=1}^\NumberOfSamples \Utility_{\PlayerIndex}(\StratProfile; \SamplePoint'_\SampleIndex) - \Utility_{\PlayerIndex}(\StratProfile; \SamplePoint_\SampleIndex)}} & \begin{tabular}{r} \textsc{Jensen's Inequality} \\ \textsc{Linearity of Expectation} \\ \end{tabular} \!\!\! \\
&= \!\! \Expwrt{\Samples,\Samples'\!,\vsigma}{\sup_{(\PlayerIndex, \StratProfile) \in \Indices}  \abs{\frac{1}{\NumberOfSamples}\sum_{\SampleIndex=1}^{\NumberOfSamples} \vsigma_\SampleIndex(\Utility_{\PlayerIndex}(\StratProfile ; \SamplePoint'_\SampleIndex) - \Utility_{\PlayerIndex}(\StratProfile; \SamplePoint_\SampleIndex))}} & \textsc{See Above} \\
&\leq \! \Expwrt{\Samples'\!, \vsigma}{\sup_{(\PlayerIndex, \StratProfile) \in \Indices}  \abs{\frac{1}{\NumberOfSamples}\sum_{\SampleIndex=1}^{\NumberOfSamples} \vsigma_\SampleIndex \Utility_{\PlayerIndex}(\StratProfile; \SamplePoint'_\SampleIndex)}} + \! \Expwrt{\Samples, \vsigma}{\sup_{(\PlayerIndex, \StratProfile) \in \Indices}  \abs{-\frac{1}{\NumberOfSamples}\sum_{\SampleIndex=1}^{\NumberOfSamples} \vsigma_\SampleIndex\Utility_{\PlayerIndex}(\StratProfile; \SamplePoint_\SampleIndex)}} \hspace{-0.25cm} & \begin{tabular}{r} \textsc{Triangle Inequality} \\ \textsc{Linearity of Expectation} \\ \end{tabular} \hspace{-0.19cm} \\ %
&= 2\RC{\NumberOfSamples}{\GameTuple, \Indices}{\ConditionDistribution} & \textsc{Definition of } \RC{\NumberOfSamples}{\GameTuple, \Indices}{\ConditionDistribution}
\end{align*}
}}\par %
\end{proof}

\begin{proof}[Proof of \Cref{thm:Rademacher}]
By \Cref{app:lemma:symmetrization}, it holds that
\[
\Expwrt{\Samples \distributed \ConditionDistribution^\NumberOfSamples}{\sup_{(\PlayerIndex, \StratProfile) \in \Indices}
\abs{\Utility_{\PlayerIndex} (\StratProfile; \ConditionDistribution) - \hat{\Utility}_{\PlayerIndex}  (\StratProfile; \Samples)}} - 2 \DERC{\NumberOfSamples}{\GameTuple, \Indices}{\vsigma}{\Samples} \le 0 \enspace ,
\]
for $\vsigma \distributed \Rademacher^\NumberOfSamples$.  Note also that, as $\Utility_{\PlayerIndex}(\StratProfile ; \SamplePoint_\SampleIndex) \in [- \nicefrac{\UtilityRange}{2}, \nicefrac{\UtilityRange}{2}]$, changing any $\SamplePoint_\SampleIndex$ results in a change of no more than $\nicefrac{3\UtilityRange}{\NumberOfSamples}$ to this difference.  Consequently, by McDiarmid's bounded difference inequality \cite{mcdiarmid1989method}, the first result holds.  
\textbf{N.B.} Standard presentations of this result, e.g. \cite{bartlett2002rademacher,mitzenmacher2017probability}, apply McDiarmid's inequality twice, on the supremum deviation and the Rademacher complexity (combining via union bound), yielding a slightly weaker bound; here we apply it only once, directly to their difference.

The second result holds similarly, though changing $\SamplePoint_\SampleIndex$ induces no change to the (non-empirical) Rademacher complexity, thus we apply the bounded difference inequality with $\nicefrac{\UtilityRange}{\NumberOfSamples}$.  The final upper bound holds via Massart's finite class inequality \citep{massart2000some}.
\end{proof}

\newcommand{\NumberOfStrategies}{\abs{\bm{S}}}

\newcommand{\del}{0.05}
\newcommand{\na}{100}
\newcommand{\ns}{10000}

\section{Rademacher Bounds for Games with Factored Noise}
\label{sec:comparingBounds}

In this appendix, we design a class of games, \mydef{games with factored noise}, for which Rademacher bounds perform well.  Crucially, this performance holds regardless of whether or not it is known \emph{a priori\/} that a game is in this class.  Consequently, Rademacher bounds are an attractive choice when we suspect a game is well-behaved, but have insufficient \emph{a priori\/} knowledge to prove any particular property that would imply strong statistical bounds, as is standard in EGTA.

The following theorem upper bounds the RA in games with factored noise.  Although our bound is rather loose, it is often substantially lower than the worst-case estimate of \Cref{thm:Rademacher}.  While computing this bound requires knowledge of the structure of the game, computing the 1-ERA does not require such knowledge; moreover, we would expect to see similar behavior when the 1-ERA is used in place of the RA, as the 1-ERA is tightly concentrated about its expectation.

\begin{theorem}[Rademacher Averages of Games with Factored Noise]
\label{thm:ra-factored}
Consider conditional normal-form game $\ConditionalGame{\ConditionSpace}$ together with distribution $\ConditionDistribution$ and index set $\smash{\Indices \subseteq \SetOfPlayers \times \StratProfileSpace}$ s.t.\ %
$\smash{\sup_{(\PlayerIndex, \StratProfile) \in \Indices} \abs{\Utility_{\PlayerIndex}(\StratProfile; \ConditionDistribution)} \leq a_0}$.  Suppose further \mydef{factoring functions} $\eta, \phi$ such that for all $\ConditionValue \in \ConditionSpace$ and $(\PlayerIndex, \StratProfile) \in \Indices$,
\[
\Utility_\PlayerIndex(\StratProfile; \ConditionValue) = \Utility_{\PlayerIndex}(\StratProfile; \ConditionDistribution) + \sum_{i=1}^n \eta_i\bigl(\phi_i(\PlayerIndex, \StratProfile); \ConditionValue\bigr) \enspace,
\]
and for all $i \in 1, \dots, n$, $\Image(\eta_i) \subseteq [-a_i, a_i]$ and $\abs{\Image(\phi_i)} \leq b_i$.  Then
\[
\RC{\NumberOfSamples}{\GameTuple, \Indices}{\rand{D}} \leq \frac{a_0}{\sqrt{\NumberOfSamples}} + \sum_{i=1}^n a_i\min\left(1, \mathsmaller{\sqrt{\frac{2\ln(b_i)}{\NumberOfSamples}}}\right) \enspace.
\]
\end{theorem}

\begin{proof}
This result follows by analyzing the contribution of each factor to the Rademacher complexity via Massart's finite class inequality (we use the form given by \citet[][Theorem~3.3]{boucheron2005theory}).  In detail:

\hspace{-0.25cm}\scalebox{0.93}{\parbox{\linewidth}{
\begin{align*}
\RC{\NumberOfSamples}{\GameTuple, \Indices}{\rand{D}} \hspace{-1.5cm} & \hspace{1.5cm} = \Expwrt{\vsigma,\Samples}{\sup_{(\PlayerIndex, \StratProfile) \in \Indices} \abs{\frac{1}{\NumberOfSamples}\sum_{\SampleIndex=1}^{\NumberOfSamples} \vsigma_\SampleIndex \Utility_\PlayerIndex(\StratProfile; \SamplePoint_\SampleIndex) }} & \textsc{Definition of $\textcomp{R}$} \\
 &= \Expwrt{\vsigma,\Samples}{\sup_{(\PlayerIndex, \StratProfile) \in \Indices} \abs{\frac{1}{\NumberOfSamples}\sum_{\SampleIndex=1}^{\NumberOfSamples} \vsigma_\SampleIndex\Utility_{\PlayerIndex}(\StratProfile,\ConditionDistribution) + \sum_{i=1}^n\vsigma_\SampleIndex \eta_i(\phi_i(\PlayerIndex, \StratProfile); \SamplePoint_j)}} & \textsc{Factoring of $\GameTuple$} \\
 &\leq \Expwrt{\vsigma,\Samples}{\sup_{(\PlayerIndex, \StratProfile) \in \Indices} \abs{\frac{1}{\NumberOfSamples}\sum_{\SampleIndex=1}^{\NumberOfSamples} \vsigma_\SampleIndex\Utility_{\PlayerIndex}(\StratProfile,\ConditionDistribution)}} + \sum_{i=1}^n \Expwrt{\vsigma,\Samples}{\sup_{(\PlayerIndex, \StratProfile) \in \Indices}\abs{\frac{1}{\NumberOfSamples}\sum_{\SampleIndex=1}^{\NumberOfSamples} \vsigma_\SampleIndex \eta_i(\phi_i(\PlayerIndex, \StratProfile); \SamplePoint_j)}} & \textsc{Triangle Inequality} \\
 &\leq \Expwrt{\vsigma,\Samples}{\sup_{(\PlayerIndex, \StratProfile) \in \Indices} \abs{\frac{1}{\NumberOfSamples}\sum_{\SampleIndex=1}^{\NumberOfSamples} \vsigma_\SampleIndex\Utility_{\PlayerIndex}(\StratProfile,\ConditionDistribution)}} + \sum_{i=1}^n a_i\min\left(1, \mathsmaller{\sqrt{\frac{2\ln(b_i)}{\NumberOfSamples}}}\right) & \hspace{-1cm}\textsc{Massart's Inequality} \\
 &\leq \frac{a_0}{\sqrt{\NumberOfSamples}} + \sum_{i=1}^{n} a_i\min\left(1, \mathsmaller{\sqrt{\frac{2\ln(b_i)}{\NumberOfSamples}}}\right) & \textsc{See Below}
\end{align*}
}}

\noindent
To see the final step, note that $\sup_{(\PlayerIndex, \StratProfile) \in \Indices} \abs{\frac{1}{\NumberOfSamples}\sum_{\SampleIndex=1}^{\NumberOfSamples} \vsigma_\SampleIndex\Utility_{\PlayerIndex}(\StratProfile; \ConditionDistribution)}$ is always realized by $(\PlayerIndex, \StratProfile)$ with the largest or smallest expectation (depending on whether $\sum_{\SampleIndex=1}^{\NumberOfSamples}\vsigma_\SampleIndex$ is positive or negative), the absolute value of which does not exceed $a_0$.  As such, the absolute value can be upper-bounded by the expected distance travelled in a unit random walk scaled by $\nicefrac{a_0}{\NumberOfSamples}$, which does not exceed $\nicefrac{a_0}{\sqrt{\NumberOfSamples}}$.
\end{proof}

As an example, we construct a game with noise that is factored into five additive components, the first three with scale 1, and the last two with scale $\frac{1}{2}$; thus $a = \langle 1, 1, 1, \frac{1}{2}, \frac{1}{2} \rangle$.  The game is enormous; the number of agents $\NumberOfPlayers$ is varied, and each agent has the same set of $\NumberOfStrategies = 100$ strategies available.  The first factor is global noise, independent of the agent or strategy profile; thus $\phi_1(\PlayerIndex, \StratProfile)$ is constant and $b_1 = 1$.  The second factor is agent noise, so ignores the strategy profile; thus $\phi_2(\PlayerIndex, \StratProfile) = \PlayerIndex$ and $b_2 = \NumberOfPlayers$.  The third factor is independent per-agent-action noise, which ignores all but the actions of a particular agent; thus $\phi_3(\PlayerIndex, \StratProfile) = \StratProfile_\PlayerIndex$ and $b_3 = \NumberOfStrategies$.  Finally, the fourth and fifth factors depend on the strategy profile (but not the agent index), and both the strategy profile and agent index, respectively; thus $b_4 = \NumberOfStrategies^{\NumberOfPlayers}$ and $b_5 = \NumberOfPlayers\NumberOfStrategies^{\NumberOfPlayers}$.  

In Figure~\ref{fig:factored-games}, we plot the Hoeffding and Rademacher bounds for this game.
We see that the former outperforms the latter for small $\NumberOfPlayers$, but near $\NumberOfPlayers = 35$, the trend reverses.  To be conservative, we plot the upper bound on the Rademacher average (derived in Theorem~\ref{thm:ra-factored}) using the constant (3) of the 1-ERA (see Theorem~\ref{thm:Rademacher}).

\begin{figure}[htbp]

\begin{multicols}{2} %

\begin{tikzpicture}

\begin{axis}[
        no markers,
        domain=1:100,
        samples=25,smooth,
        legend pos=south east,
        legend style={draw=none,fill=none,font=\scriptsize},
        legend cell align={left},
        width=3.3in,
        height=2.6in,
        xticklabel style={
            font=\small
        },
        yticklabel style={
            /pgf/number format/fixed,
            /pgf/number format/precision=3,
            font=\small
        },
        scaled y ticks=false,
        xlabel={\small{}Number of Agents $\NumberOfPlayers$},
        xmin=0,xmax=100,
        ymin=0,ymax=1.6,
        ytick distance=0.5,minor y tick num=4,
        xtick distance=10,minor x tick num=1,
        grid=both,
        grid style={line width=.4pt,opacity=0.18},major grid style={line width=.5pt,opacity=0.55},
    ]

\addplot[style=solid,color=blue,domain=1:18]{10 * sqrt((ln(2) + ln(x) + x * ln(\na) - ln(\del)) / (2 * \ns))};
\addplot[style=solid,color=blue,domain=18:100,forget plot=true]{10 * sqrt((ln(2) + ln(x) + x * ln(\na) - ln(\del)) / (2 * \ns))};
\addlegendentry{Hoeffding: $\norm{\bm{a}}_{\!1} \sqrt{\frac{2\ln \bigl( \nicefrac{2\NumberOfPlayers \NumberOfStrategies^{\NumberOfPlayers}}{\delta} \bigr)}{m}}$};

\addplot[style=solid,color=red,domain=1:12]{2 * (1 / sqrt(\ns) + sqrt(2 * ln(2) / \ns) + sqrt(2 * ln(2 * \na) / \ns) + sqrt(2 * ln(2 * x) / \ns) + 0.5 * sqrt(2 * x * ln(2 * \na) / \ns) + 0.5 * sqrt(2 * (ln(2 * x) + x * ln(\na)) / \ns)) + 10 * 3 * sqrt(-ln(\del) / (2 * \ns))};
\addplot[style=solid,color=red,domain=12:100,forget plot=true,samples=15]{2 * (1 / sqrt(\ns) + sqrt(2 * ln(2) / \ns) + sqrt(2 * ln(2 * \na) / \ns) + sqrt(2 * ln(2 * x) / \ns) + 0.5 * sqrt(2 * x * ln(2 * \na) / \ns) + 0.5 * sqrt(2 * (ln(2 * x) + x * ln(\na)) / \ns)) + 10 * 3 * sqrt(-ln(\del) / (2 * \ns))};
\addlegendentry{Rademacher: $2\RC{m}{\GameTuple, \SetOfPlayers \times \StratProfileSpace}{\rand{D}} + 3\norm{\bm{a}}_{\!1} \sqrt{\frac{2\ln \left( \nicefrac{1}{\delta} \right)}{m}}$};

\end{axis}
\end{tikzpicture}

\vfill\columnbreak

\begin{align*}
\bm{a} &= \langle 1, 1, 1, \mathsmaller{\frac{1}{2}}, \mathsmaller{\frac{1}{2}} \rangle & \\
\bm{b} &= \langle 1, \NumberOfPlayers, \NumberOfStrategies, \NumberOfStrategies^{\NumberOfPlayers}, \NumberOfPlayers \NumberOfStrategies^{\NumberOfPlayers} \rangle & \\
\end{align*}

\vspace{-0.5cm}

\begin{align*}
\ \ \ \ \ \ \ \Utility_\PlayerIndex(\StratProfile; \SamplePoint) &= \Utility_{\PlayerIndex}(\StratProfile; \ConditionDistribution) + \mathsmaller{\sum_{i=1}^5} \eta_i(\phi_i(\StratProfile, \PlayerIndex); \SamplePoint) & \\
 &= \Utility_{\PlayerIndex}(\StratProfile; \ConditionDistribution) + \eta_1(\phi_1; \SamplePoint) & \\
 & \ \ \ + \eta_2(\phi_2(p); \SamplePoint) + \eta_3(\phi_3(\StratProfile_p); \SamplePoint) & \\
 & \ \ \ + \mathsmaller{\frac{1}{2}}\eta_4(\phi_4(\StratProfile); \SamplePoint) + \mathsmaller{\frac{1}{2}}\eta_5(\phi_5(\StratProfile, \PlayerIndex); \SamplePoint) & \\
\end{align*}

\vfill
\end{multicols}

\vspace{-0.25cm}

\caption{Hoeffding vs.\ Rademacher Bounds for a game with factored noise with (variable) $\NumberOfPlayers$ agents, $\NumberOfStrategies = \na$ actions per agent, $\NumberOfSamples = \na$ samples, and $\delta = \del$.}
\label{fig:factored-games}
\end{figure}

\end{document}